\documentclass{article}
\usepackage[letterpaper, margin=1.0in]{geometry}
\usepackage{graphicx} 
\usepackage{amsmath, amssymb, amsthm}
\usepackage{hyperref}
\usepackage{algorithm}
\usepackage{algpseudocode}
\usepackage{caption}
\usepackage{subcaption}
\usepackage{xcolor}
\usepackage{natbib}
\usepackage{tikz}
\title{  1.64-Approximation for Chromatic Correlation Clustering via Chromatic Cluster LP}
\newtheorem{theorem}{Theorem}
\newtheorem{lemma}{Lemma}
\newtheorem{definition}{Definition}
\newtheorem{proposition}{Proposition}

\newtheorem{claim}{Claim}
\DeclareMathOperator{\obj}{obj}
\DeclareMathOperator{\opt}{opt}

\DeclareMathOperator{\err}{err}

\newcommand{\acks}{\section*{Acknowledgments}}
\begin{document}
\date{}
\author{
  Dahoon Lee\thanks{Department of Mathematical Sciences, Seoul National University, Seoul, 08826, South Korea. \texttt{dahoon46@snu.ac.kr}} \and
  Chenglin Fan\thanks{ Department of Computer Science and Engineering, Seoul National University, Seoul, 08826, South Korea. \texttt{fanchenglin@snu.ac.kr}  } \and
  Euiwoong Lee\thanks{Computer Science and Engineering Division, University of Michigan, Ann Arbor, MI 48109, USA. \texttt{euiwoong@umich.edu}}
}
\maketitle

\begin{abstract}

Chromatic Correlation Clustering (CCC)
generalizes Correlation Clustering by assigning multiple categorical relationships (colors) to edges and imposing chromatic constraints on the clusters. Unlike traditional Correlation Clustering, which only deals with binary $(+/-)$ relationships, CCC captures richer relational structures.
Despite its importance, improving the approximation for CCC has been difficult due to the limitations of standard LP relaxations.

We present a randomized $1.64$-approximation algorithm to the CCC problem, significantly improving the previous factor of $2.15$.
Our approach extends the cluster LP framework
to the chromatic setting by introducing a chromatic cluster LP relaxation and an rounding algorithm that utilizes both a cluster-based and a greedy pivot-based strategy. The analysis bypasses the integrality gap of $2$ for the CCC version of standard LP and highlights the potential of the cluster LP framework to address other variants of clustering problems.

\end{abstract}

\section{Introduction}

\textbf{Correlation Clustering (CC)} is a foundational unsupervised learning problem that partitions a signed graph while minimizing disagreements with the provided edge labels. Given an undirected graph $G = (V, E = E^+ \uplus E^-)$, where edges are labeled as either similar (`+') or dissimilar (`$-$'), the goal is to find a node partition (namely, a clustering) $\mathcal{C} = \{C_1, \ldots, C_k\}\:(\biguplus_{i=1}^k C_i=V)$ that minimizes the total number of disagreements:

\begin{equation}
\text{cost}(\mathcal{C}) =  \underbrace{|\{uv \in E^+ : u \in C_i, v \in C_j, i \neq j\}|}_{\text{Positive edges between clusters}}+\underbrace{|\{uv \in E^- : \exists C_i, u,v \in C_i \}|}_{\text{Negative edges within a cluster}}.
\label{eq:disagreement}
\end{equation}

The first term counts the number of similar pairs that are incorrectly separated, while the second counts dissimilar pairs that are incorrectly grouped together~\citep{bansal2004correlation}.
\\

To find approximate solutions, CC is often formulated as a linear program (LP). The standard \textbf{edge-based LP relaxation} uses fractional variables $x_{uv} \in [0,1]$ to represent the probability that nodes $u$ and $v$ are separated. This formulation enforces triangle inequality constraints ($x_{uw} \leq x_{uv} + x_{vw}$) to ensure the separation variables induce a valid pseudometric, while minimizing the objective:
\begin{align}
\text{minimize}\quad&\sum_{(u,v) \in E^+} x_{uv} + \sum_{(u,v) \in E^-} (1 - x_{uv})\tag{LP}\end{align}
\begin{align}
\text{subject to}\quad &x_{uv} + x_{vw} \geq x_{wu}& \forall u,v,w\in V \\
    &x_{uv} \in [0, 1]& \forall uv\in\binom{V}{2}
\end{align}

A stronger formulation is the \textbf{cluster LP} relaxation~\citep{Cao2024}, which assigns every subset $S$ of $V$ to variable $z_S$.
Since the number of variables in the LP is exponential to $|V|$, \citep{Cao2024} construct the near-optimal solution in polynomial time.
\\

While standard CC excels at capturing binary relationships, many real-world applications involve richer, multi-categorical data. For instance, social networks encode diverse interaction types (e.g., professional, familial), and biological networks annotate proteins with distinct functional labels. In such settings, clusters must satisfy both structural coherence and categorical homogeneity. This motivates \textbf{Chromatic Correlation Clustering (CCC)}~\citep{bonchi2012chromatic}, a powerful generalization of CC.

CCC operates on edge-colored graphs, where each color represents a distinct relationship type. The problem imposes two key requirements on the final clustering:
\begin{itemize}
    \item \textbf{Chromatic Clusters}: Each cluster is now assigned a color (or label) as well.
    \item \textbf{Color-Consistent Edges}: Edges \emph{within} a cluster should ideally match the cluster's assigned color, while edges \emph{between} different clusters should ideally be negative ($=\gamma$) color, indicating dissimilarity.
\end{itemize}

The objective in CCC is to find a partition that minimizes the number of disagreements—edges that violate these conditions.
Formally, given the color assignment $\phi:\binom{V}{2}\to L\cup \{\gamma\}\:(|L|<\infty)$,
the goal is to find a node partition $\mathcal{C}$ along with the color assignment $\Phi:\mathcal{C}\to L$ that minimizes the following:
\begin{align}
\text{cost}((\mathcal{C},\Phi)) &=  \underbrace{|\{uv : \phi(uv)\in L,(u \in C_i, v \in C_j, i \neq j\text{ or }\exists C_i, u,v \in C_i,\Phi(C_i)\neq \phi(uv)|}_{\text{Positive edges with violated assignment}}\nonumber\\
&\quad+\underbrace{|\{uv : \phi(uv)=\gamma,\exists C_i, u,v \in C_i \}|}_{\text{Negative edges within any cluster}}.
\label{eq:disagreement_CCC}
\end{align}
This makes CCC NP-hard yet highly valuable for applications such as detecting communities with uniform relationship types in network analysis~\citep{bonchi2012chromatic}, merging records in databases under categorical constraints~\citep{anava2015chromatic}, and clustering proteins with shared functional annotations~\citep{klodt2021color}.


\subsection{Our Contributions}
We present a significant improvement over previous theoretical works on CCC by extending the cluster LP framework to the chromatic setting.
\begin{theorem}[Main result]
\label{thm:final}
    For any $\varepsilon>0$, there exists a $(\frac{18}{11}+\varepsilon)$-approximation algorithm for CCC with time complexity $n^{\textnormal{poly}(1/\varepsilon)}\cdot \textnormal{poly}(L)$.
\end{theorem}
This improves upon the previous best-known approximation factor of 2.15~\citep{Lee2025} for this problem.

Our contributions on top of the main result are the following:
\begin{itemize}

    \item \textbf{Flexibility on the Extension of the Cluster LP Framework:} We extend the powerful cluster LP framework by \citep{Cao2024} to the chromatic setting, maintaining intuition for the framework. This demonstrates the robustness and generality of the cluster LP approach.

    \item \textbf{A Strengthened Standard LP Relaxation for CCC:} 
    Derived from the extension of cluster LP to the chromatic setting, we also provide that the standard LP for CCC can be further strengthened in a natural way (see Condition~\eqref{cond:CCC_metric_strong} compared to Condition~\eqref{cond:CCC_metric}).
    \item \textbf{Supplementary Intuitions on Cluster LP Framework:} We provide more intuitions on some parts of the cluster LP framework, particularly on the novel constraint (Constraint~\eqref{cond:LP_strong}) of the revised version of cluster LP and the policy of choosing between cluster-based and pivot-based algorithms from the given data.


\end{itemize}

\subsection{Overview of Techniques}
\paragraph{Extension of Cluster LP.}
In Section~\ref{sec:LP}, we will extend the following \eqref{eq:CLP} to the chromatic setting.
\begin{align}
\text{minimize}\quad&\sum_{uv\in E^+} x_{uv} + \sum_{uv\in E^-}(1-x_{uv})\label{eq:CLP}\tag{cluster LP}\end{align}
\begin{align}
\text{subject to}\quad
&\sum_{S\ni u}z_S = 1& \forall u\in V\label{cond:LP_v}\\
&\sum_{S\ni uv}z_S = 1-x_{uv}& \forall uv\in \binom{V}{2}\label{cond:LP_e}\\
&\sum_{|S\cap uvw|\geq 2}z_S\leq 1& \forall uvw\in \binom{V}{3}\label{cond:LP_strong}\\
&z_S\geq 0& \forall S\subseteq V,\,S\neq \emptyset\label{cond:LP_bd}
\end{align}

Here, $z_S\in[0,1]$  indicates the number of cluster $S$ within the clustering.
Therefore, constraint \eqref{cond:LP_v} ensures that each node is contained in some unique cluster, while constraint \eqref{cond:LP_e} ensures that the separation variables $x_{uv}$ are consistent with the probability of nodes $u$ and $v$ belong to the same cluster.

Constraint~\eqref{cond:LP_strong} is an additional, yet crucial constraint for the cluster LP originally suggested, required to bring the approximation factor for the CC problem down to $1.437+\varepsilon$ from previous papers~\citep{cohenaddad2023preclustering,Cao2024,Cao2025}.
Intuitively, this constraint indicates that there are at most 1 cluster containing at least 2 vertices among $u,\,v,\,$and $w$.

These ideas can be well preserved in the extension to the chromatic setting by presenting the decomposition of $z_S$ into variables $\{z_S^c\}_{c\in L}$ (see Section~\ref{sec:LP}).

\paragraph{Extension of Near-Optimal Solution Construction.}
In Section~\ref{sec:solution}, we will also make a generalization of the whole procedure suggested by \citep{Cao2024} of constructing the near-optimal solution for the cluster LP towards the chromatic setting.
The following describes the high-level idea on each of the procedures with our extension conducted:
\begin{enumerate}
    \item \textbf{Preclustering -- Section~\ref{subsec:preclust}:}
    Starting from clustering such that $\obj=O(1)\cdot \opt$, the preclustering step identifies \emph{locally near-optimal} (i.e., mostly $+$ internally and mostly $-$ externally) disjoint non-singleton subsets (denoted as \emph{atoms}) and which pairs of atom (or singleton) are strongly far apart (denoted as non-\emph{admissible edges})~\citep{Cao2025}.

    In the chromatic setting, each atom and admissible edge now have a color: The color of atom should be unique, while multiple edges with different colors between two vertices are possible.
    Since every atom is completely inside some cluster with the same color of optimal clustering and the target near-optimal clustering due to Lemma~\ref{lem:atom}, admissible edges incident on some atom are only allowed to have the same color as the atom.
    \item \textbf{Bounded Sub-Cluster LP Relaxation -- Section~\ref{subsec:bsclp}:}
    Based on the information about atoms and admissible edges from the previous step, \emph{bounded sub-cluster LP} is then formulated to find the $(1+\varepsilon)$-approximate clustering.
    The bounded sub-cluster LP consists of variables $y_S^s$, indicating the number of clusters with size $s$ containing $S$.
    By restricting the upper bound of $s$ as the polynomial of $1/\varepsilon$, the number of variables in the LP is then polynomial by $n$, thus solving the LP takes $n^{\text{poly}(1/\varepsilon)}$ time.

    In the chromatic setting, the natural extension for the variable is then $y_S^{s,c}$, indicating the number of clusters with size $s$ of color $c$ containing $S$.
    Since the number of variable is then multiplied by $L$, solving the LP takes $n^{\text{poly}(1/\varepsilon)}\cdot \text{poly}(L)$ time.
    \item \textbf{Conversion to the Solution for Cluster LP -- Section~\ref{subsec:sampling} and \ref{subsec:clp}:}
    The optimal solution for the bounded sub-cluster LP is then converted to the solution for the original cluster LP, based on the fact that the cluster inside the near-optimal solution can be effectively sampled from $\{y_S^s\}$ using the Raghavendra-Tan rounding technique~\citep{RT12}.
    Following the sampling of $(s,u)$ with a probability proportional to $y_u^s$, the technique is applied to $\{y_{Su}^s\}$ normalized by $y_u^s$.

    In the chromatic setting, the $c$ from $y_S^{s,c}$ is yet another categorical property just as $s$; therefore, the process can be generalized by sampling $(s,c,u)$ with a probability proportional to $y_u^{s,c}$, followed by the utilization of the Raghavendra-Tan rounding technique on $\{y_{Su}^{s,c}\}$ normalized by $y_u^{s,c}$.
\end{enumerate}

\paragraph{Extension of Clustering Algorithm with Triangle Analysis.}
In Section~\ref{sec:algorithm},
we introduce a $\frac{18}{11}$-approximation algorithm for CCC from the solution of the chromatic cluster LP, which would be constructed from the previous step.
The algorithm is based on the idea suggested by \citep{Cao2024}, which randomly selects the execution between the cluster-based algorithm and the pivot-based algorithm.

The original cluster-based algorithm samples the cluster $S$ with a probability proportional to $z_S$, which can easily be generalized to the chromatic setting by sampling $(S,c)$ with a probability proportional to $z_S^c$.

The pivot-based algorithm used in the algorithm of \citep{Cao2024} not only uses the specific deterministic rounding function but also exploits more information from cluster LP solution by sampling the cluster during the pivoted subroutine.
In this work, we will instead use a purely pivot-based algorithm for CCC with the \emph{color pre-categorization}~\citep{XiuHTCH22} and \emph{rounding function}~\citep{chawla2015near,Lee2025} technique.
Although the color pre-categorization technique allows for simpler analysis of the algorithm, this has the limitation that the approximation factor cannot be less than $1.5$, which would be discussed in Section~\ref{sec:triangle}.

We choose rounding functions for the pivot-based algorithm that correspond to the greedy selection on the pivot edges, which then yields the promising $\frac{18}{11}$-approximation in total.
This can be proven by the usual triangle analysis for CCC, with the budget non-linearly depending on the fractional edge separation value $x$ of the LP.
This is because the performance of the cluster-based algorithm depends on the value of $x$, which is between $[1,2]$~\citep{Cao2024}.
The proof of such an approximation factor is given in Section~\ref{sec:triangle}.

Based on this result, we would provide an intuition in Section~\ref{sec:algorithm} about when we would prefer to choose the pivot-based algorithm over the cluster-based algorithm, depending on the given input $\phi$.
Also, since the rounding functions were naively chosen in the pivot-base algorithm, we believe that an even better approximation factor can be achieved by fine-tuning the rounding functions, or even further by exploiting more information from the cluster LP solution as in the pivot-based algorithm by \citep{Cao2024}.

\subsection{Related Work}
The \emph{Correlation Clustering} (CC) problem, introduced in~\citep{ben-dor1999clustering}, is APX-hard with progressive approximation improvements. Early work established constant-factor approximations~\citep{bansal2004correlation}, later refined to 4-approximation~\citep{CHARIKAR2005360} and 3-approximation~\citep{ailon2008aggregating}. Subsequent LP-rounding techniques broke the 2-approximation barrier~\citep{chawla2015near,cohenaddad2022sherali}, with recent advances using Sherali-Adams hierarchies~\citep{cohenaddad2022sherali} and preclustering~\citep{cohenaddad2023preclustering} achieving $(1.73+\varepsilon)$. The current best $(1.437+\varepsilon)$-approximation comes from Cao et al.'s \emph{cluster LP} framework~\citep{Cao2024}, which unifies previous relaxations and admits efficient solutions.
Recent work by \citep{Cao2025} presents a breakthrough method for solving the Correlation Clustering LP in sublinear time.

Multi-relational network analysis motivated \emph{Chromatic Correlation Clustering} (CCC), which extends traditional correlation clustering via colored edges. While early work \citep{bonchi2012chromatic} developed heuristics, theoretical foundations emerged with \citep{anava2015chromatic}'s 4-approximation LP method. \citep{klodt2021color} later connected CCC to classical results, showing Pivot maintains a 3-approximation.
Recent advances in approximation algorithms for CCC have achieved new theoretical milestones. \citep{XiuHTCH22} developed an LP-based method that improved the approximation factor to 2.5. \citep{Lee2025} advanced the state-of-the-art by developing a 2.15-approximation algorithm for Chromatic Correlation Clustering through an innovative combination of standard LP relaxation with problem-specific rounding techniques. Their work also established an inherent limitation of this approach, proving that no approximation ratio better than 2.11 can be achieved within the same framework. In contemporary data analysis, the application of correlation clustering frequently encounters practical limitations, such as restricted memory or the need to process data as it arrives in a continuous stream. This growing demand for efficiency has catalyzed extensive research into the development of specialized correlation clustering algorithms designed for dynamic, streaming, online, and distributed environments~\citep{lattanzi2017consistent,  jaghargh2019, cohen2019fully, lattanzi2021parallel,guo2021distributed,fichtenberger2021, cohen2022online, assadi2022sublinear, behnezhad2022, behnezhad2023, bateni2023,cohen2024dynamic, abs-2402-15668, braverman2025fully,abs-2504-12060}.


\section{Overview of the Framework}
\subsection{Chromatic Cluster LP}
\label{sec:LP}
The standard LP relaxation for CCC suggested by \citep{bonchi2012chromatic} is as follows:
\begin{align}
    \text{minimize}\quad& \sum_{\phi(uv)\neq \gamma} x_{uv}^{\phi(uv)} + \sum_{\phi(uv) = \gamma} \sum_{c \in L} (1 - x_{uv}^c) \label{eq:LP_CCC} \tag{CCC-LP} \end{align}
\begin{align}
    \text{subject to}\quad 
    &x_{uv}^c \geq x_u^c,\, x_v^c&\forall uv\in \binom{V}{2},\,c\in L \label{cond:CCC_evdom} \\
    &x_{uv}^c + x_{vw}^c \geq x_{wu}^c&\forall u,v,w\in V,\,c\in L \label{cond:CCC_metric} \\
    &\sum_{c \in L} x_u^c = |L| - 1&\forall u\in V \label{cond:CCC_chroma} \\
    &x_u^c,\,x_v^c,\, x_{uv}^c \in [0,1]&\forall uv\in \binom{V}{2},\,c\in L \label{cond:CCC_bd}
\end{align}
Here, each $1-x_u^c$ and $1-x_{uv}^c$ represents the fractional assignment of the vertex $u$ and edge $uv$ to a cluster of color $c$, respectively.

Applying the idea of cluster LP, we can formulate the chromatic version of the cluster LP as \eqref{eq:CLP-CCC}.
We generalize the primary version of \eqref{eq:CLP}, which does not include constraint~\eqref{cond:LP_strong}.
\begin{align}
\text{minimize}\quad&\sum_{\phi(uv)\in L} x_{uv}^{\phi(uv)} + \sum_{\phi(uv)=\gamma}\sum_{c\in L} (1-x_{uv}^c)\label{eq:CLP-CCC}\tag{chromatic cluster LP}\end{align}
\begin{align}
\text{subject to}\quad
&\sum_{S\ni u}z_S^c = 1-x_u^c&\forall u\in V,\, c\in L\label{cond:CLP_vc}\\
&\sum_{c\in L}(1-x_u^c)=1&\forall u\in V\label{cond:CLP_v}\\
&\sum_{S\ni uv}z_S^c = 1-x_{uv}^c&\forall uv\in \binom{V}{2},\,c\in L\label{cond:CLP_ec}\\
&z_S^c\geq 0& \forall S\subseteq V,\,S\neq \emptyset,\,c\in L\label{cond:CLP_bd}
\end{align}
Here, $z_S^c$ indicates the fractional count of cluster $S$ of color $c$ within the clustering.
Therefore, constraint~\eqref{cond:CLP_v} ensures that each node is contained in some unique cluster,
while constraint~\eqref{cond:CLP_vc} and \eqref{cond:CLP_ec} ensures that $x_u^c$ and $x_{uv}^c$ are consistent with the probability that node $u$ belongs to the cluster of color $c$, or nodes $u$ and $v$ belong to the same cluster of color $c$, respectively.

We can also check that the solution from the cluster LP induces the feasible solution for the standard LP: \eqref{cond:CCC_evdom} can be deduced from \eqref{cond:CLP_vc} and \eqref{cond:CLP_ec}, 
\eqref{cond:CCC_chroma} can be deduced from \eqref{cond:CLP_v},
and \eqref{cond:CCC_bd} can be deduced from all 4 inequality constraints --- \eqref{cond:CLP_vc} to \eqref{cond:CLP_bd} --- in \eqref{eq:CLP-CCC}.

When it comes to the condition~\eqref{cond:CCC_metric}, even stronger constraint can be deduced, which is the following:
\begin{equation}
\label{cond:CCC_metric_strong}
x_{uv}^c + x_{vw}^c \geq x_{wu}^c+x_v^c\quad \forall uvw\in \binom{V}{3},\,c\in L.
\end{equation}
In fact, this constraint is natural regarding the interpretation of variables $x_{uv}^c$ as \emph{bottleneck distances}~\citep{Lee2025}:
Since $x_{uv}^c$ is a bottleneck distance $d(u,v)$ from $u$ to $v$ conditioned by the use of edges $uu_c$ and $vv_c$,
\begin{align*}
    x_{uv}^c + x_{vw}^c&= \max\{d(u,u_c),d(u_c,v_c),d(v_c,v)\}+\max\{d(v,v_c),d(v_c,w_c),d(w_c,w)\}\\
    &=\max\{\max\{d(u,u_c),d(u_c,v_c)\},d(v_c,v)\}\\
    &\quad+\max\{d(v,v_c),\max\{d(v_c,w_c),d(w_c,w)\}\}\\
    &\geq d(v_c,v)+\max\{\max\{d(u,u_c),d(u_c,v_c)\},\max\{d(v_c,w_c),d(w_c,w)\}\}\\
    &\geq d(v_c,v)+\max\{d(u,u_c),d(u_c,w_c),d(w_c,w)\}\\
    &=x_v^c+x_{uw}^c,
\end{align*}
where the first inequality comes from the fact that $\max\{a,c\}+\max\{b,c\}\geq c+\max\{a,b\}$.

Since the number of variables in the chromatic cluster LP is $|L|\cdot 2^{|V|}$, we obtain the near-optimal solution rather than the actual optimal solution for the LP, whose procedure is described in Section~\ref{sec:solution}.
Thus, we have the following theorem:
\begin{theorem}
\label{thm:solution}
    The feasible solution $\{z_S^c\}$ for $\eqref{eq:CLP-CCC}$ such that $\obj(z)=(1+\varepsilon)\opt$ can be obtained in $n^{\textnormal{poly}(1/\varepsilon)}\cdot\textnormal{poly}(L)$ time.
\end{theorem}
\subsection{Algorithmic Framework}
\label{sec:algorithm}
We can naturally extend the original cluster-based rounding algorithm by~\citep{Cao2024} to the chromatic setting, which is described in Algorithm~\ref{alg:alg1}.
\begin{algorithm}
\caption{Cluster-Based Rounding}
\label{alg:alg1}
\begin{algorithmic}
    \State \textbf{Input:} Near-optimal \eqref{eq:CLP-CCC} solution $\{z_S^c\}$
    \State \textbf{Output:} Chromatic clustering $\Phi$\State
    \State $\Phi(V^2)\leftarrow \gamma,\,V'\leftarrow V$
    \While{$V'\neq \emptyset$}
    \State Choose $(c,S)$ with probability proportional to $z_S^c$ over all $c\in L,\,\emptyset\neq S\subseteq V$
    \If{$V'\cap S\neq \emptyset$}
    \State $\Phi((V'\cap S)^2)\leftarrow c,\,V'\leftarrow V'\backslash S$
    \EndIf
    \EndWhile
\end{algorithmic}
\end{algorithm}

Using Algorithm~\ref{alg:alg1} exclusively results in $2+\varepsilon$-approximation algorithm, given that the solution to \eqref{eq:CLP-CCC} is $(1+O(\varepsilon))$-optimal.
\begin{theorem}
    Algorithm~\ref{alg:alg1} yields a $(2+\varepsilon)$-approximation algorithm for CCC with time complexity $n^{\textnormal{poly}(1/\varepsilon)}\cdot\textnormal{poly}(L)$.
\end{theorem}
\begin{proof}
    In Algorithm~\ref{alg:alg1}, for every $uv\in \binom{V}{2}$, the probability of vertices being separated is $\frac{\sum_{c\in L}(1-x_{uv}^c)}{2-\sum_{c\in L}(1-x_{uv}^c)}\leq \sum_{c\in L}(1-x_{uv}^c)$, while the probability of vertices not being in the same cluster of color $c'$ is $1-\frac{1-x_{uv}^{c'}}{2-\sum_{c\in L}(1-x_{uv}^c)}=\frac{2x_{uv}^{c'}-\sum_{c\neq c'}(1-x_{uv}^c)}{2-\sum_{c\in L}(1-x_{uv}^c)}\leq \frac{2x_{uv}^{c'}}{1+x_{uv}^{c'}}$.
    Therefore, the probability of violation for the edge $uv$ is bounded above by twice the value of the corresponding term in the objective.
    Since $\obj(\{z_S^c\})=(1+O(\varepsilon))\opt$, this finalizes the proof. 
\end{proof}
According to the proof, the bottleneck of the approximation factor occurs on the $+$ edge whose $x$-value is close to $0$.
This behavior is different from the usual pivot-based algorithm, whose bottleneck usually occurs when some $x$-value is large with the tight triangle inequality.
Thus, this factor can be moderated by randomizing the execution of algorithm between cluster-based algorithm (Algorithm~\ref{alg:alg1}) and pivot-based algorithm (Algorithm~\ref{alg:alg2}~\citep{XiuHTCH22, Lee2025}.), which is formally described in Algorithm~\ref{alg:mix}~\citep{Cao2024}.

Applying the derandomization, the algorithm would choose the better option out of two:
Intuitively, the algorithm would choose the pivot-based algorithm if the input $\phi$ is \emph{almost well-clustered} (i.e., clusters of the optimal clustering consist of a very low number of violating internal edges), and choose cluster-based algorithm otherwise.
\begin{algorithm}
\caption{Pivot-Based Rounding}
\label{alg:alg2}
\begin{algorithmic}
    \State \textbf{Input:} Near-optimal \eqref{eq:CLP-CCC} solution $\{z_S^c\}$
    \State \textbf{Output:} Chromatic clustering $\Phi$\State
    \State $\Phi(V^2)\leftarrow \gamma,\,V_c\leftarrow \emptyset$ for all $c\in L$
    \For{$u\in V$}
    \If{$\exists c\in L$ s.t.\ $x_u^c<\frac{1}{2}$}
        \State $V_c\leftarrow V_c\cup\{u\}$
    \EndIf
    \EndFor
    \For{$c\in L$}
    \While{$V_c\neq \emptyset$}
    \State Choose $u\in V_c$ uniformly random
    \State $C\leftarrow \{u\}$
    \For{$v\in V_c\backslash\{u\}$}
        \State Set $p_{uv}$ as following:
        \[p_{uv}=\begin{cases}
            f^+(x_{uv}),&\phi(uv)=c;\\
            f^-(x_{uv}),&\phi(uv)=\gamma;\\
            f^\circ(x_{uv}),&\text{otherwise}.
        \end{cases}\]
        \State $C\leftarrow C\cup\{v\}$ with probability $1-p_{uv}$
    \EndFor
    \State $\Phi(C^2)\leftarrow c,\,V_c\leftarrow V_c\backslash C$
    \EndWhile
    \EndFor
\end{algorithmic}
\end{algorithm}
\begin{algorithm}
    \caption{Cluster/Pivot-based Rounding}
    \label{alg:mix}
    \begin{algorithmic}
    \State \textbf{Input:} Near-optimal \eqref{eq:CLP-CCC} solution $\{z_S^c\}$, target approximation factor $\alpha\in [1.5,2]$
    \State \textbf{Output:} Chromatic clustering $\Phi$\State\State
    \Return{Algorithm~\ref{alg:alg1} with probability $\alpha/2$, Algorithm~\ref{alg:alg2} otherwise}
    \end{algorithmic}
\end{algorithm}

Indeed, Algorithm~\ref{alg:mix} achieves the significantly smaller approximation factor of $\frac{18}{11}+\varepsilon\approx 1.64+\varepsilon$, even with greedy selection of pivoted edges (formally, $f^+=0,\,f^-=f^\circ=1$).
The proof of the factor is provided in Section~\ref{sec:triangle}.

\subsection{Clustering as a Binary Function}
For analysis purposes, we provide an alternative definition for chromatic clustering.
\begin{proposition}[Chromatic Clustering as a Binary Function]
    Chromatic Clustering $\Phi:\mathcal{C}\to L$ has a one-to-one correspondence with
    the symmetric function $\Phi:V\times V\to L\cup\{0,\gamma\}$ such that
    \begin{enumerate}
    \item $\Phi(u,v)=0\iff u=v,$
    \item $\Phi(u,v)\in L\implies \Phi(u,w)=\Phi(v,w)\in\{\Phi(u,v),\gamma\},\,\forall w\in V\backslash\{u,v\}.$
    \end{enumerate}
\end{proposition}
This interpretation is analogous to the ultrametric from the UMVD problem.
When comparing the order structure of the codomain between CCC and UMVD, as shown in Figure~\ref{fig:order},
elements of the codomain for the CCC problem are organized horizontally,
while those for the UMVD problem are organized vertically.
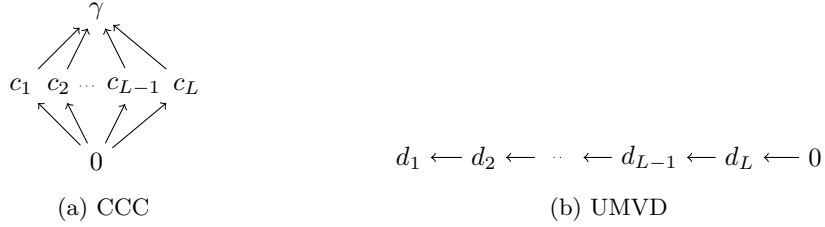
\begin{figure}
    \centering
    \begin{subfigure}[b]{0.4\textwidth}
        \centering
        \begin{tikzpicture}
    \node (n) at (-1.2, 1) {};
    \node (a) at (0,0) {$\gamma$};
    \node (b1) at (-1,-1) {$c_1$};
    \node (b2) at (-0.5,-1) {$c_2$};
    \node (b3) at (0.5,-1) {$c_{L-1}$};
    \node (b4) at (1.2,-1) {$c_L$};
    \node (c) at (0,-2) {$0$};
    \draw[<-] (a) -- (b1);
    \draw[<-] (a) -- (b2);
    \draw[<-] (a) -- (b3);
    \draw[<-] (a) -- (b4);
    \draw[->] (c) -- (b1);
    \draw[->] (c) -- (b2);
    \draw[->] (c) -- (b3);
    \draw[->] (c) -- (b4);
    \draw[dotted] (b2) -- (b3);
\end{tikzpicture}
        \caption{CCC}
    \end{subfigure}
    \begin{subfigure}[b]{0.4\textwidth}
        \centering\begin{tikzpicture}
    \node (a) at (-2,0) {$d_1$};
    \node (b) at (-1,0) {$d_2$};
    \node (c) at (1.2,0) {$d_{L-1}$};
    \node (d) at (2.4,0) {$d_L$};
    \node (e) at (3.4,0) {$0$};
    \node (b1) at (-0.2,0) {};
    \node (c1) at (0.2,0) {};
    \draw[<-] (a) -- (b);
    \draw[<-] (b) -- (b1);
    \draw[dotted] (b1) -- (c1);
    \draw[<-] (c1) -- (c);
    \draw[<-] (c) -- (d);
    \draw[<-] (d) -- (e);
\end{tikzpicture}
        \caption{UMVD}
    \end{subfigure}
    
    \caption{Comparison in order structure of codomain for two problems.
    The directional edge $a\to b$ indicates $a\leq b$. Self-loop on each vertex is omitted.}
    \label{fig:order}
\end{figure}

Next, we provide some terminology helpful for analyzing local cost or its difference.
\begin{definition}
    For a vertex $v$, define $\Phi_v:=\Phi(v,\cdot)$.
\end{definition}
\begin{definition}
    For $f,g:V\to R$ and $K\subseteq V$, define $d_0(f,g):=|\{v\in V\mid f(v)\neq g(v)\}|$ and $d_0^K(f,g):=d_0(f,g)+|\{v\in V\backslash K\mid f(v)\neq g(v)\}|$.
\end{definition}
Note that the weight given for the vertices outside $K$ is twice larger than for those inside $K$ in the definition of $d_0^K$.
This terminology is especially helpful when analyzing the difference in objective value between two clusterings whose difference occurs only when at least one of the vertex in the pair is included in $K$.
\begin{proposition}
\label{prop:obj}
    $\obj(\Phi)=\frac{1}{2}\sum_{u\in V}d_0(\phi_u,\Phi_u)$.
\end{proposition}
\begin{proposition}
\label{prop:obj_diff}
    If $\{(u,v)\mid \Phi^1(u,v)\neq \Phi^2(u,v)\}\subseteq K\times V+V\times K$, then
    \[\obj(\Phi^1)-\obj(\Phi^2)=\frac{1}{2}\sum_{u\in K}d_0^K(\phi_u,\Phi^1_{u})-d_0^K(\phi_u,\Phi^2_{u}).\]
\end{proposition}
\begin{proposition}
\label{prop:d0diff}
    $d_0(f,g)-d_0(f,h)=|\{v\in V\mid g(v)\neq f(v)=h(v)\}|-|\{v\in V\mid g(v)= f(v)\neq h(v)\}|$.
\end{proposition}
The following definitions provide alternative definitions of basic terminologies on chromatic clustering from the perspective of a binary function.
\begin{definition}
    The cluster $C$ of $\Phi$, denoted as $C\in \Phi$, is a maximal subset of $V$ such that $C^2\subseteq \Phi^{-1}(L\cup \{0\})$.
\end{definition}
\begin{definition}
    If $S\subseteq \exists C\in \Phi$ for a nonempty subset $S$ of $V$ (that is, a subset of some cluster), we define $\Phi(S)$ as $\max \Phi(C^2)$.
\end{definition}

\section{Triangle Analysis}
\label{sec:triangle}
We first prove that Algorithm~\ref{alg:mix} achieves an approximation factor of $(\frac{18}{11}+\varepsilon)$.
Recall our selection of rounding functions:
\begin{equation}
    \label{eq:round}f^+(x)=0,\quad f^-(x)=f^\circ (x) = 1.
\end{equation}
The goal is to ensure that the Algorithm~\ref{alg:mix} with the target approximation factor $\alpha$ actually achieves the $\alpha+\varepsilon$-approximation if $\alpha\geq \frac{18}{11}$.

Let any $uv\in \binom{V}{2}$ be given.
If $\not\exists c\in L$ such that $u,v\in V_c$, then $\Phi(u,v)=\gamma$ if Algorithm~\ref{alg:alg2} is used, due to the first subroutine of the algorithm.
On the other hand, any of the values $x_{uv}^c$ is bounded below by $1/2$, since either $x_u^c$ or $x_v^c$ is bounded below by $1/2$.
Therefore, if $\phi(uv)\in L$, then the cost incurred on the edge $uv$ is
\[\frac{\alpha}{2}\cdot \frac{2x_{uv}^{\phi(uv)}-\sum_{c\neq \phi(uv)}(1-x_{uv}^c)}{2-\sum_{c\in L}(1-x_{uv}^c)}+\left(1-\frac{\alpha}{2}\right)\cdot 1
\leq \frac{\alpha}{2}\cdot\frac{2x_{uv}^{\phi(uv)}}{1+x_{uv}^{\phi(uv)}}+1-\frac{\alpha}{2}=1-\frac{\alpha}{2}\cdot\frac{1-x_{uv}^{\phi(uv)}}{1+x_{uv}^{\phi(uv)}},\]
whose difference with $\alpha x_{uv}^{\phi(uv)}$ is
\[\alpha x_{uv}^{\phi(uv)}-\left(1-\frac{\alpha}{2}\cdot\frac{1-x_{uv}^{\phi(uv)}}{1+x_{uv}^{\phi(uv)}}\right)=\frac{\alpha}{2}\cdot \left(2(1+x_{uv}^{\phi(uv)})+\frac{2}{1+x_{uv}^{\phi(uv)}}-3\right)-1\geq \frac{\alpha}{2}\cdot\frac{4}{3}-1\geq 0,\]
since $x_{uv}^{\phi(uv)}\geq 1/2>0$ and $\alpha \geq 1.5$.
Thus, the cost incurred on the edge is bounded above by $\alpha$ times the corresponding LP value.
Since the inequality is tight, this also implies that Algorithm~\ref{alg:mix} cannot achieve the approximation factor below $1.5$ due to the strict majority selection in the first subroutine of Algorithm~\ref{alg:alg2}.

If $\phi(uv)=\gamma$, then the cost incurred on the edge $uv$ is
\[\frac{\alpha}{2}\cdot\frac{\sum_{c\in L}(1-x_{uv}^c)}{2-\sum_{c\in L}(1-x_{uv}^c)}+\left(1-\frac{\alpha}{2}\right)\cdot 0\leq \frac{\alpha}{2}\cdot\sum_{c\in L}(1-x_{uv}^c)\leq \alpha\cdot \sum_{c\in L}(1-x_{uv}^c),\]
which also implies that the cost incurred on the edge is bounded above by $\alpha$ times the corresponding LP value.

Therefore, it is sufficient to check on edges $uv$ where both $x_u^c$ and $x_v^c$ are smaller than $1/2$ for some $c\in L$, i.e., the vertex $u$ and $v$ have the same strict majority color as $c$.
We can apply the usual triangle analysis to analyze Algorithm~\ref{alg:alg2} since the second subroutine of the algorithm is exactly the color-wise execution of the pivot algorithm~\citep{XiuHTCH22,Lee2025}.

Now, fix any $c\in L$.
\begin{definition}[Edge sign assignment with respect to color $c$]
    $E^+=\phi^{-1}(c)\cap \binom{V_c}{2},\,E^-=\phi^{-1}(\gamma)\cap \binom{V_c}{2},\,E^\circ=\phi^{-1}(L\backslash \{c\})\cap \binom{V_c}{2}$.
\end{definition}

Suppose:
\begin{equation}\label{eq:triangle}
    e.cost_w(uv)+e.cost_u(vw)+e.cost_v(wu)\leq e.lp_w(uv)+ e.lp_u(vw)+ e.lp_v(wu)
\end{equation}
for all $u,v,w\in V_c$, where
\begin{equation}
e.cost_w(u,v) = \begin{cases}
    p_{uw}(1 - p_{vw}) + (1 - p_{uw})p_{vw}, & uv \in E^+; \\
    (1 - p_{uw})(1 - p_{vw}), & uv \in E^-; \\
    1 - p_{uw}p_{vw}, & uv \in E^\circ;
\end{cases}
\end{equation}
\begin{equation}\label{eq:ccc_elp}
e.lp_w(u,v) = (1 - p_{uw}p_{vw})\cdot b(uv),
\end{equation}
\begin{equation}\label{eq:ccc_b}
b(uv) = \begin{cases}
    \frac{\alpha}{1-\alpha/2}\cdot\frac{{x_{uv}^c}^2}{1+x_{uv}^c}, & uv \in E^+; \\
    \frac{\alpha}{1-\alpha/2}\cdot\frac{(1+2x_{uv}^c)(1 - x_{uv}^c)}{2(1+x_{uv}^c)}, & uv \in E^-; \\
    \frac{\alpha}{1-\alpha/2}\cdot\frac{\max\left\{\frac{1}{2}, 1 - x_{uv}^c, 1 - x_{vw}^c, 1 - x_{wu}^c \right\}^2}{1+\max\left\{\frac{1}{2}, 1 - x_{uv}^c, 1 - x_{vw}^c, 1 - x_{wu}^c \right\}}, & uv \in E^\circ.
\end{cases}
\end{equation}
Here, $e.cost_w(u,v)$ denotes the probability of a violation on the edge $uv$ conditioned on the pivot $w$, while $e.lp_w(u,v)$ denotes the probability of at least one of $u$ and $v$ contained in the cluster $C$ conditioned on the pivot $w$ multiplied by the charging factor $b$.
Then the following is the sufficient condition to achieve the approximation factor of $\alpha+\varepsilon$~\citep{Cao2024}:
\begin{align*}
    \phi(uv)\in L&\implies\frac{\alpha}{2}\cdot \frac{2x_{uv}^{\phi(uv)}-\sum_{c\neq \phi(uv)}(1-x_{uv}^c)}{2-\sum_{c\in L}(1-x_{uv}^c)}+\left(1-\frac{\alpha}{2}\right)\cdot b(uv)
\leq \alpha x_{uv}^{\phi(uv)},\\
\phi(uv)\in \gamma&\implies \frac{\alpha}{2}\cdot\frac{\sum_{c\in L}(1-x_{uv}^c)}{2-\sum_{c\in L}(1-x_{uv}^c)}+\left(1-\frac{\alpha}{2}\right)\cdot b(uv)\leq \alpha\cdot \sum_{c\in L}(1-x_{uv}^c),
\end{align*}
which turns out to be true.
\begin{lemma}
    If inequality~\eqref{eq:triangle} holds for all possible configurations, then Algorithm~\ref{alg:mix} yields an $\alpha+\varepsilon$-approximation algorithm for CCC.
\end{lemma}
\begin{proof}
    Consider the following three cases:
    \begin{enumerate}
        \item $\phi(uv)=c$
        \begin{align*}
            \frac{\alpha}{2}\cdot \frac{2x_{uv}^{\phi(uv)}-\sum_{c'\neq \phi(uv)}(1-x_{uv}^{c'})}{2-\sum_{c'\in L}(1-x_{uv}^{c'})}+\left(1-\frac{\alpha}{2}\right)\cdot b(uv)
            \leq \frac{\alpha}{2}\cdot\frac{2x_{uv}^{c}}{1+x_{uv}^{c}}+\left(1-\frac{\alpha}{2}\right)\cdot b(uv)=\alpha x_{uv}^{c}.
        \end{align*}
        \item $\phi(uv)=\gamma$
        
        Let $h(x)=\frac{(3-2x)x}{2(2-x)}\:(x\in [0,1])$.
        Then $h$ is an increasing function since $h'(x)=\frac{(1-x)(3-x)}{(2-x)^2}\geq 0$ in the domain.
        Since $0\leq 1-x_{uv}^c\leq \sum_{c'\in L}(1-x_{uv}^{c'})\leq 1$, 
        \[\frac{1-\alpha/2}{\alpha}\cdot b(uv)=h(1-x_{uv}^c)\leq h\left(\sum_{c'\in L}(1-x_{uv}^{c'})\right)=\frac{(3-2\cdot \sum_{c'\in L}(1-x_{uv}^{c'}))\cdot \sum_{c'\in L}(1-x_{uv}^{c'})}{2(2-\sum_{c'\in L}(1-x_{uv}^{c'}))}.\]
        Therefore,
        \begin{align*}
            \frac{\alpha}{2}\cdot\frac{\sum_{c\in L}(1-x_{uv}^c)}{2-\sum_{c\in L}(1-x_{uv}^c)}+\left(1-\frac{\alpha}{2}\right)\cdot b(uv)&\leq \frac{\alpha}{2}\cdot \frac{4-2\cdot \sum_{c'\in L}(1-x_{uv}^{c'})}{2-\sum_{c'\in L}(1-x_{uv}^{c'})}\cdot \sum_{c'\in L}(1-x_{uv}^{c'})\\
            &= \alpha\cdot \sum_{c'\in L}(1-x_{uv}^{c'}).
        \end{align*}
        
        \item $\phi(uv)\in L\backslash \{c\}$

        Let $g(x)=\frac{x^2}{1+x}\:(x\in [0,1])$.
        Then $g$ is an increasing function since $g'(x)=\frac{x(2+x)}{(1+x)^2}\geq 0$ in the domain.
        Since $0\leq \max\left\{\frac{1}{2}, 1 - x_{uv}^c, 1 - x_{vw}^c, 1 - x_{wu}^c \right\}\leq x_{uv}^{\phi(uv)}$~\citep{XiuHTCH22},
        \[\frac{1-\alpha/2}{\alpha}\cdot b(uv)=g\left(\max\left\{\frac{1}{2}, 1 - x_{uv}^c, 1 - x_{vw}^c, 1 - x_{wu}^c \right\}\right)\leq g(x_{uv}^{\phi(uv)})=\frac{{x_{uv}^{\phi(uv)}}^2}{1+x_{uv}^{\phi(uv)}}.\]
        Therefore,
        \begin{align*}
            \frac{\alpha}{2}\cdot \frac{2x_{uv}^{\phi(uv)}-\sum_{c'\neq \phi(uv)}(1-x_{uv}^{c'})}{2-\sum_{c'\in L}(1-x_{uv}^{c'})}+\left(1-\frac{\alpha}{2}\right)\cdot b(uv)
            &\leq \alpha\cdot \frac{x_{uv}^{\phi(uv)}+{x_{uv}^{\phi(uv)}}^2}{1+x_{uv}^{\phi(uv)}}\\
            &=\alpha\cdot x_{uv}^{\phi(uv)}.
        \end{align*}
    \end{enumerate}
    Therefore, the sufficient condition for the approximation factor of $\alpha+\varepsilon$ is satisfied.
\end{proof}

Now, it is sufficient to show that the inequality~\eqref{eq:triangle} is valid.
\begin{definition}
    \[ALG(uvw):=e.cost_w(uv)+e.cost_u(vw)+e.cost_v(wu),\]
    \[LP(uvw):=e.lp_w(uv)+ e.lp_u(vw)+ e.lp_v(wu).\]
\end{definition}
\begin{lemma}
    The ineqaulity~\eqref{eq:triangle} is true with rounding functions given as \eqref{eq:round} and $\alpha=\frac{18}{11}\approx 1.64$.
\end{lemma}
\begin{proof}
    We first provide a proof for some trivial cases.
    \begin{claim}
        If $ALG=0$, then the inequality is valid.
    \end{claim}
    \begin{proof}
        Since $b(uv)\geq 0$ and $(1-p_{uw}p_{vw})\geq 0$, 
        $e.lp_w(uv)\geq 0$.
        Likewise, other two terms are also nonnegative, hence $LP\geq 0=ALG$.
    \end{proof}
    We can check that $ALG=0$ unless the configuration of the sign is one of $(+,+,-),\,(+,+,\circ),\,(+,-,\circ),\,(+,\circ,\circ)$, since the value of the rounding function $f$ is constant---either 0 or 1---whenever the sign is fixed.
    Therefore, it is now sufficient to check on the remaining 4 cases.

    Define $(x,y,z):=(x_{uv},x_{vw},x_{wu})$. Recall that $h(x)=\frac{(3-2x)x}{2(2-x)}$ and $g(x)=\frac{x^2}{x+1}$ are increasing functions, hence $h(1-x)=\frac{(1+2x)(1-x)}{2(1+x)}$ is a decreasing function.
    \paragraph{$(+,+,-)$ Triangles}
    The values of $ALG$ and $LP$ are as follows:
    \[ALG=3,\quad LP=9\cdot\left(\frac{x^2}{1+x}+\frac{y^2}{1+y}+\frac{(1+2z)(1-z)}{2(1+z)}\right).\]
    We first minimize $LP$ when $x+y$ is fixed.
    The function $g$ is convex since $g''(x)=\frac{2}{(1+x)^3}\geq 0$.
    Therefore, $g(x)+g(y)\geq 2g\left(\frac{x+y}{2}\right)$, indicating that $LP$ is minimized at $x=y$.

    Now set $y=x$. Since $h(1-x)$ is a decreasing function, $LP$ is minimized at $z=\min\{2x,1\}$.
    There are two cases with respect to the value of $x$.
    \begin{enumerate}
        \item $x<\frac{1}{2}$
        
        $LP$ is minimized at $(x,x,2x)$. Therefore,
        \begin{align*}
            LP\geq 9\cdot (2g(x)+h(1-2x))=9\cdot \left(\frac{(1-2x)(1+5x)}{6(1+x)(1+2x)}+\frac{1}{3}\right)\geq 3 = ALG.
        \end{align*}
        \item $x\geq \frac{1}{2}$

        $LP$ is minimized at $(x,x,1)$. Therefore,
        \begin{align*}
            LP\geq 9\cdot 2\cdot g(x)\geq 9\cdot 2\cdot\frac{1}{6}=3=ALG.
        \end{align*}
    \end{enumerate}
    \paragraph{$(+,+,\circ)$ Triangles}
    The values of $ALG$ and $LP$ are as follows:
    \[ALG=3,\quad LP=9\cdot\left(\frac{x^2}{1+x}+\frac{y^2}{1+y}+\frac{\max\{\frac{1}{2},1-x,1-y,1-z\}^2}{1+\max\{\frac{1}{2},1-x,1-y,1-z\}}\right).\]
    As in the previous case, set $y=x$ since $LP$ is minimized at $x=y$ by the convexity of $g$.
    Since $g$ is an increasing function, $LP$ is minimized when $\max\{\frac{1}{2},1-x,1-y,1-z\}$ is minimized, thus at $z=\min\{2x,1\}$.

    In summary, $LP$ is minimized at $(x,x,\min\{2x,1\})$ with a value
    \[LP\geq 9\cdot \left(2g(x)+g\left(\max\left\{\frac{1}{2},1-x,1-x,1-2x\right\}\right)\right)=9\cdot \left(2g(x)+g\left(\max\left\{\frac{1}{2},1-x\right\}\right)\right).\]
    There are two cases with respect to the value of $x$.
    \begin{enumerate}
        \item $x<\frac{1}{2}$
        \begin{align*}
            LP\geq 9\cdot (2g(x)+g(1-x))=9\cdot \frac{1-x+3x^2-x^3}{(1+x)(2-x)}\geq 3 = ALG.
        \end{align*}
        To check the last inequality,
        \[3(1-x+3x^2-x^3)-(1+x)(2-x)=1-4x+10x^2-3x^3\]
        is minimized at $x=2/9$ with value $139/243\geq 0$,
        since the derivative is $-(2-x)(2-9x)$.
        \item $x\geq \frac{1}{2}$
        \begin{align*}
            LP\geq 9\cdot (2g(x)+g(1/2))\geq 9\cdot3\cdot \frac{1}{6}=\frac{9}{2}\geq 3 = ALG.
        \end{align*}
    \end{enumerate}
    \paragraph{$(+,-,\circ)$ Triangles}
    The values of $ALG$ and $LP$ are as follows:
    \[ALG=1,\quad LP=9\cdot\left(\frac{x^2}{1+x}+\frac{(1+2y)(1-y)}{2(1+y)}+\frac{\max\{\frac{1}{2},1-x,1-y,1-z\}^2}{1+\max\left\{\frac{1}{2},1-x,1-y,1-z\right\}}\right).\]
    Since $g$ is an increasing function,
    \begin{align*}
        LP\geq 9\cdot \frac{1/4}{1+1/2}=1.5 \geq 1 = ALG.
    \end{align*}
    \paragraph{$(+,\circ,\circ)$ Triangles}
    The values of $ALG$ and $LP$ are as follows:
    \[ALG=2,\quad LP=9\cdot\left(\frac{x^2}{1+x}+2\cdot \frac{\max\{\frac{1}{2},1-x,1-y,1-z\}^2}{1+\max\{\frac{1}{2},1-x,1-y,1-z\}}\right).\]
    As in the previous case,
    \begin{align*}
        LP \geq 9\cdot 2 \cdot \frac{1/4}{1+1/2} = 3 \geq 2 = ALG.
    \end{align*}

    This completes our case analysis.
\end{proof}

This finalizes the proof of Theorem~\ref{thm:final}.
\section{Construction of Near-Optimal Solution to Chromatic Cluster LP}
\label{sec:solution}
In this section, we describe the whole procedure of constructing the near-optimal solution to \eqref{eq:CLP-CCC}.
\subsection{Preclustering}
\label{subsec:preclust}
We first provide a chromatic version of preclustering, stated in the following Theorem:
\begin{theorem}[Chromatic Extension of Theorem 14 of \citep{Cao2024}]
\label{thm:preclustering}
    There exists an algorithm that generates $(\Phi^p,\,E^2)$ such that
    \begin{itemize}
        \item $\exists \Phi_1^*$ s.t.\ $\obj(\Phi_1^*)=(1+\varepsilon)\opt\,\Phi_1^*(u,v)=c\in L\implies uv\in E^{2;c}$, and
        \item $|E^2|\leq O\left(\frac{1}{\varepsilon^2}\right)\cdot \opt$.
    \end{itemize}
    in $O(\textnormal{poly}(n,\frac{1}{\varepsilon},L))$ time.
\end{theorem}
The intuition of the atom for the CCC problem is compatible with those for the original CC problem: It is almost a clique with a single color.
The construction of atoms, described in Algorithm~\ref{alg:atom}, also follows in a compatible way in that only vertices with small local cost are intact.
Therefore, the objective value for the output of the algorithm is scaled up to $O(1/\beta^2)$ times the original.
Thus, the following lemma holds since $\beta=0.1$ is a small constant and the input of the algorithm is $O(1)-$approximation clustering.
\begin{algorithm}
\caption{Constructing atoms}
\label{alg:atom}
\begin{algorithmic}
\State \textbf{Input:} $O(1)$ approximation clustering $\Phi$
    \State \textbf{Output:} Clustering by atoms $\Phi^p$
    \State\For{$C\in\Phi$}
    \For{$u\in C$}
    \If{$d_0(\phi_u,\Phi_u)>\frac{\beta}{2}\cdot |C|$} mark $u$
    \EndIf
    \EndFor
    \If{at least $\frac{\beta}{3}$ ratio of $C$ is marked} mark $C$
    \EndIf
    \EndFor
    \State \[\Phi^p(u,v)=\begin{cases}
    \gamma,&u\text{ or }v\text{ is marked},\,u\neq v;\\
    \Phi(u,v),&\text{otherwise.}
    \end{cases}\]
\end{algorithmic}
\end{algorithm}
\begin{lemma}[Chromatic Extension of Lemma 66 of \citep{Cao2024}]
    $\obj(\Phi^p)\leq O(1)\cdot \opt$.
\end{lemma}
Next, we prove that not only the optimal clustering can be constructed without breaking atoms, but also the color of non-singleton atoms stays intact.
\begin{lemma}[Chromatic Extension of Lemma 67 of \citep{Cao2024}]
\label{lem:atom-1}
    If $K\in \Phi^p$ is non-singleton, then $d_0(\phi_u,\Phi_u^p)<\beta|K|$ for all $u\in K$.
\end{lemma}
\begin{proof}
    Since the size of the original cluster $C\in \Phi$ containing $K$ is at most $|K|/(1-\beta/3)$,
    \[d_0(\phi_u,\Phi_u^p)\leq d_0(\phi_u,\Phi_u)+|C\backslash K|\leq \frac{\beta}{2}\cdot |C|+\frac{\beta}{3}\cdot |C|\leq \frac{5\beta}{6}\cdot \frac{|K|}{1-\beta/3}<\beta |K|.\]
\end{proof}
\begin{lemma}[Chromatic Extension of  Lemma 68 of \citep{Cao2024}]
    \label{lem:atom}
    For any optimal cluster $\Phi^*$, $\Phi^*\leq \Phi^p$.
\end{lemma}
\begin{proof}
    If the atom $K\in\Phi^p$ is singleton (i.e., $K=\{u\}$), since $\Phi^p_u=\gamma \cdot 1_{V\backslash\{u\}}$, $\Phi^*_u\leq \Phi^p_u$.

    Now suppose $K$ is non-singleton. If $|K\cap C|<\frac{2}{3}|K|$ for all $C\in \Phi^*$, consider constructing a cluster $K$ from $\Phi^*$; formally,
    \[\Phi'(u,v)=\begin{cases}
        \Phi^p(u,v),&u,v\in K,\,u\neq v;\\
        \gamma,&(u\in K,\,v\not\in K)\text{ or }(u\not\in K,\,v\in K);\\
        \Phi^*(u,v),&\text{otherwise}.
    \end{cases}\]
    Then for all $u\in K$,
    \[
        d_0^K(\phi_u,\Phi^*_u)-d_0^K(\phi_u,\Phi'_u)\geq \frac{|K|}{3}-2\beta |K| >0
    \]
    due to Proposition~\ref{prop:d0diff}, Lemma~\ref{lem:atom-1}, and the definition of $d_0^K$.
    Therefore, by Proposition~\ref{prop:obj_diff},
    \[\obj(\Phi^*)-\obj(\Phi')=\frac{1}{2}\sum_{u\in K}d_0^K(\phi_u,\Phi^*_u)-d_0^K(\phi_u,\Phi'_u)>0,\]
    which contradicts with the optimality of $\Phi^*$.

    Now set $C\in\Phi^*$ as a unique cluster satisfying $|K\cap C|\geq \frac{2}{3}|K|$. If $\Phi(C)\neq \Phi^p(K)$, consider splitting $K\cap C$ from $C$ and assigning the color $\Phi^p(K)$ to it; formally,
    \[\Phi'(u,v)=\begin{cases}
        \Phi^p(u,v),&u,v\in K\cap C,\,u\neq v;\\
        \gamma,&(u\in K\cap C,\,v\not\in K\cap C)\text{ or }(u\not\in K\cap C,\,v\in K\cap C);\\
        \Phi^*(u,v),&\text{otherwise}.
    \end{cases}\]
    Then for all $u\in K\cap C$,
    \[
     d_0^{K\cap C}(\phi_u,\Phi^*_u)-d_0^{K\cap C}(\phi_u,\Phi'_u)\geq \frac{2}{3}|K|-2\beta |K| >0,
    \]
    using similar arguments as above. Again, this contradicts with the optimality of $\Phi^*$. Thus, $\Phi(C)= \Phi^p(K)$.


    Finally, suppose $K\backslash C\neq \emptyset$.
    Let $c=|C\backslash K|$. Then the number of edges between $K\cap C$ and $C\backslash K$ of color $\Phi^*(C)$ is at least $\frac{c}{2}|K\cap C|$,
    since otherwise splitting $K\cap C$ from $C$ is more optimal.
    By Lemma~\ref{lem:atom-1}, the number of edges between $K\cap C$ and $K\backslash C$ of color $\neq \Phi^*(C)$ is at most $\left(\beta|K| -\frac{c}{2}\right)|K\cap C|$.
    Thus, we can find $k\in K\backslash C$ whose number of edges between $K\cap C$ of color $\neq \Phi^*(C)$ is at most $\min\left\{\beta|K|,\left(\beta|K| -\frac{c}{2}\right)\cdot \frac{|K\cap C|}{|K\backslash C|}\right\}$.
    Consider moving $k$ from its cluster in $\Phi^*$ to $C$; formally,
    \[
    \Phi'(u,v)=\begin{cases}
        \Phi^*(C),&(u=k,\,v\in C)\text{ or }(u\in C,\, v=k);\\
        \gamma,&(u=k,\,v\not\in C\cup\{k\})\text{ or }(u\not\in C\cup\{k\},\, v=k);\\
        \Phi^*(u,v),&\text{otherwise}.
    \end{cases}
    \]
    Then
    \begin{align*}
    d_0^{\{k\}}(\phi_k,\Phi^*_k)-d_0^{\{k\}}(\phi_k,\Phi'_k)&\geq 2\left(|K\cap C|-|K\backslash C|-c-2\min\left\{\beta|K|,\left(\beta|K| -\frac{c}{2}\right)\cdot \frac{|K\cap C|}{|K\backslash C|}\right\}\right.\\
    &\quad\left.-\left(\beta|K|-\min\left\{\beta|K|,\left(\beta|K| -\frac{c}{2}\right)\cdot \frac{|K\cap C|}{|K\backslash C|}\right\}\right)\right)\\
    &=2\left(|K\cap C|-|K\backslash C|-\beta|K|-c-\min\left\{\beta|K|,\left(\beta|K| -\frac{c}{2}\right)\cdot \frac{|K\cap C|}{|K\backslash C|}\right\}\right)\\
    &\geq 2\left(|K\cap C|-|K\backslash C|-2\beta|K|-\frac{2\beta(|K\cap C|-|K\backslash C|)}{|K\cap C|}\cdot |K|\right)\\
    &\geq 2\left(\frac{1}{3}-3\beta\right)|K|>0
    \end{align*}
    when $\beta=0.1$, where the equality condition for the last inequality is $|K\cap C|=\frac{2}{3}|K|$.
    This contradicts with the optimality of $\Phi^*$. Thus, $K\backslash C=\emptyset$ and $K\subseteq C$, providing $\Phi^*\leq \Phi^p$ since $\Phi^*(u,v)=\Phi^p(u,v)$ whenever $\Phi^p(u,v)<\gamma$.
\end{proof}
We denote $K_u$ as an atom containing the vertex $u$, as in the original paper.
We also adopt the notion of normalized degree.
\begin{definition}
    $w_{uv}^c:=\frac{|\phi^{-1}(c)\cap (K_u\times K_v)|}{|K_u||K_v|},\,w_{u,C}^c:=\sum_{v\in C}w_{uv}^c,\,w_u^c:=w_{u,V\backslash K_u}^c+|K_u|\cdot 1_{\Phi^p(K_u)\leq c}$.
\end{definition}

Now we construct admissible edges. Admissible edges also have a color; hence the analysis can be done color-wise.
Consider Algorithm~\ref{alg:phi1}, which continuously splits the atom from the cluster if such an operation is not relatively expensive, starting from the optimal clustering.
Here, $c\cdot \chi_S:V\rightarrow L\cup\{0,\gamma \}$ is defined as
\[c\cdot \chi_S(v)=\begin{cases}
    c,&v\in S;\\
    \gamma,&\text{otherwise.}
\end{cases}\]
Since $\obj(\Phi_1^*)-\obj(\Phi^*)\leq O(\varepsilon)\cdot \obj(\Phi_1^*)$ after the algorithm, $\Phi_1^*$ satisfies $\obj(\Phi_1^*)=(1+O(\varepsilon))\cdot \opt$ as well as the condition.
\begin{algorithm}
\caption{Conceptual Construction of Near-Optimal Clustering - 1}
\label{alg:phi1}
\begin{algorithmic}
    \State \textbf{Input:} Optimal clustering $\Phi^*$
    \State \textbf{Output:} Near-optimal clustering $\Phi_1^*$ with cost $(1+O(\varepsilon))\opt$\State
    \State $\Phi_1^*\leftarrow \Phi^*$
    \While{$\exists K\in \Phi^p,\,C\in\Phi_1^*$ s.t.\ $K\subsetneq C,\,\newline
    \sum_{u\in K}(d_0^K(\phi_u,\Phi_1^*(C)\cdot \chi_{K})-d_0^K(\phi_u,\Phi_1^*(C)\cdot \chi_{C}))\leq 2\varepsilon\cdot\sum_{u\in K}d_0^K(\phi_u,\Phi_1^*(C)\cdot \chi_{K})$}
    \State $\Phi_1^*(K\times C\backslash K+ C\backslash K\times K)\leftarrow \gamma$
    \EndWhile
\end{algorithmic}
\end{algorithm}

The construction of $E^2$ is stated in the following definition.
\begin{definition}
    \begin{align}
    E^{1;c}&:=\{uv\mid w_u^c>\varepsilon w_v^c\},\,E^1:=\biguplus_{c\in L}E^{1;c};\\
    E^{2;c}&:=\left\{uv\left| K_u\neq K_v,\,\sum_{p\in N_u^{1;c}\cap p\in N_v^{1;c}}w_{up}^cw_{vp}^c>\varepsilon(w_u^c+w_v^c)\right.\right\},\,E^2:=\biguplus_{c\in L}E^{2;c}.
    \end{align}
\end{definition}
\begin{lemma}[Chromatic Extension of  Lemma 70 of \citep{Cao2024}]
    If $\Phi_1^*(u,v)=c\neq \gamma$, then $w_u^c>\varepsilon w_v^c$.
\end{lemma}
\begin{proof}
    If $K_u=K_v$, then the statement is trivial.

    Otherwise, i.e.,\ $K_u\neq K_v$, there exists $C\in \Phi_1^*$ such that $K_u,K_v\subseteq C$.
    By Algorithm~\ref{alg:phi1}, the following inequality holds.
    \begin{align}\sum_{x\in K_u}(d_0^{K_u}(\phi_x,\Phi_1^*(C)\cdot \chi_{K_u})-d_0^{K_u}(\phi_x,\Phi_1^*(C)\cdot \chi_{C}))> 2\varepsilon\cdot\sum_{x\in K_u}d_0^{K_u}(\phi_x,\Phi_1^*(C)\cdot \chi_{K_u}).\label{eq:phi1}\end{align}

    Each side of the inequality can be bounded in terms of the normalized degree of color $c$ as follows.
    \begin{align}
         \sum_{x\in K_u}d_0^{K_u}(\phi_x,\Phi_1^*(C)\cdot \chi_{K_u})&\geq 2|K_u|(w_u^c-|K_u|)\label{eq:phi1_lb1}\\
         &\geq 0\label{eq:phi1_lb2},
    \end{align}
    \begin{align}
        \sum_{x\in K_u}(d_0^{K_u}(\phi_x,\Phi_1^*(C)\cdot \chi_{K_u})-d_0^{K_u}(\phi_x,\Phi_1^*(C)\cdot \chi_{C}))&\leq 4|K_u|w_{u,C\backslash K_u}^c-2|K_u||C\backslash K_u|\label{eq:phi1_ub1}\\
        &\leq 2|K_u||C\backslash K_u|\label{eq:phi1_ub2}.
    \end{align}
    All of the above are also true if every $K_u$ is replaced by $K_v$.
    Combining inequalities~\eqref{eq:phi1}, \eqref{eq:phi1_lb1} and \eqref{eq:phi1_ub2} on $v$ results in
    \[2\varepsilon(w_v^c-|K_v|)<|C\backslash K_v|,\quad 2\varepsilon w_v^c<|C|,\]
    whereas combining inequalities~\eqref{eq:phi1}, \eqref{eq:phi1_lb2} and \eqref{eq:phi1_ub1} on $u$ results in
    \[0\leq 2w_{u,C\backslash K_u}^c-|C\backslash K_u|,\quad |C|\leq 2w_u^c.\]
    Combining final two inequalities completes the proof.
\end{proof}

All the following arguments can be driven by applying the corresponding lemmas from the original paper, where the edges of color $c$ act as $+$ edges while others act as $-$ edges.
This is possible as the entire argument is based solely on the notion of $+$ degree, the degree of other colors does not affect the correctness of the argument.
Moreover, since Lemma~\ref{lem:atom} and Algorithm~\ref{alg:phi1} implies $\Phi^p(K)=\Phi^*(K)=\Phi^*_1(K)=\Phi^*_1(C)$ if $K\subseteq C\in \Phi_1^*$ and $K$ is non-singleton,
any non-singleton atoms whose color differs from $c$ do not need any admissible edges of color $c$, allowing the excision of such atoms during the analysis with color $c$.
Therefore, proofs for following lemmas are omitted.
\begin{lemma}[Chromatic Extension of  Lemma 72 of \citep{Cao2024}]
    If $\Phi_1^*(u,v)=c\in L$ and $v\not \in K_u$, then $\sum_{p\in N_u^{1;c}\cap  N_v^{1;c}}w_{up}^cw_{vp}^c>\varepsilon(w_u^c+w_v^c)$.
\end{lemma}
\begin{lemma}[Chromatic Extension of  Lemma 74 of \citep{Cao2024}]
    $|N_u^{2;c}|-|K_u|\leq O\left(\frac{1}{\varepsilon^2}\right)(w_u^c-|K_u|)$.    
\end{lemma}
Since
\[|E^2|=\frac{1}{2}\sum_{u\in V}\sum_{c\in L}(|N_u^{2;c}|-|K_u|)\leq O\left(\frac{1}{\varepsilon^2}\right)\sum_{u\in V}\sum_{c\in L}\sum_{w_u^c-k_u}=O\left(\frac{1}{\varepsilon^2}\right)\cdot \obj(\Phi^p)\leq O\left(\frac{1}{\varepsilon^2}\right)\cdot \opt,\]
this finalizes the proof for Theorem~\ref{thm:preclustering}.
\subsection{Bounded Sub-Cluster LP Relaxation}
\label{subsec:bsclp}
Next, we provide a chromatic version of bounded sub-cluster LP relaxation.
The LP consists of variables $y_S^{s,c}$ that refer to the fractional count of clusters with size $s$, color $c$, containing $S$.
As in the original paper, $s$ and the size of $|S|$ are bounded above by $r=\Theta\left(\frac{1}{\varepsilon^{12}}\right)$. Let $\varepsilon_1=\varepsilon^3$.
\begin{algorithm}
\caption{Conceptual Construction of Near-Optimal Clustering - 2}
\label{alg:phi1_1}
\begin{algorithmic}
\While{$\exists u,\,C\in \Phi_1^*$ s.t.\ $K_u\subsetneq C,\, |K_u|<|C|<|K_u|+\varepsilon_1\cdot |N_u^{2;c}\backslash K_u|$}
\State $\Phi_1^*(K_u\times C\backslash K_u+ C\backslash K_u\times K_u)\leftarrow \gamma$
\EndWhile
\end{algorithmic}
\end{algorithm}
\begin{align}
\text{minimize}\quad&\sum_{\phi(uv)\in L} x_{uv}^{\phi(uv)} +\sum_{\phi(uv)=\gamma}y_{uv}\label{eq:sLP_CCC}\tag{bounded chromatic sub-cluster LP}\end{align}
\begin{align}
\text{subject to}\quad&y_u^c=1-x_u^c&\forall u\in V,\,c\in L\\
&y_u=1&\forall u\in V\\
&y_{uv}^c=1-x_{uv}^c&\forall uv\in \binom{V}{2},\,c\in L\\
&\frac{1}{s}\sum_{u\in V} y_{Su}^{s,c}=y_S^{s,c}&\forall S\subseteq V,\,|S|\leq r,\,s\leq r,\,c\in L\label{cond:sLP_CCC_trans}\\
&y_S^{s,c} \geq 0&\forall S\subseteq V,\,|S|\leq r,\,s\leq r,\,c\in L\\
& y_{uv}^{\Phi^p(K)}=1&\forall u,v\in K\in \Phi^p,\,uv\in \binom{V}{2}\\
& y_{uv}^c=0&c\in L,\,uv\not\in E^{2;c}\\
& y_u^{s,c}=0&\forall u\in V,\,c\in L,\nonumber\\&&s\neq |K_u|,s\leq |K_u|+\varepsilon_1|N_u^{2;c}\backslash K_u|\label{eq:size}\\
&\sum_{T'\subseteq T}(-1)^{|T'|}\cdot y_{S\cup T'}^{s,c}
\in[0,y_S^{s,c}]&\forall s\leq r,\,c\in L,\,S\cap T=\emptyset
\end{align}
Here, any of the omitted superscripts in $y$ indicates the sum of $y_S^{s,c}$ over those superscripts.
All of the constraints appropriately encapsulate the idea of the bounded sub-cluster LP with color.
Since the additional cost incurred after Algorithm~\ref{alg:phi1_1} is at most $O(\varepsilon_1)\cdot |E^2|=O(\varepsilon)\cdot \opt$,
the existence of the near-optimal chromatic clustering $\Phi_1^*$ with cost $(1+O(\varepsilon))\opt$ satisfying inequality~\eqref{eq:size} is still guaranteed.
The running time for solving this LP is $L^{O(1)}n^{O(r)}$.
\subsection{Cluster Sampling}
\label{subsec:sampling}
Using the solution for the bounded sub-cluster LP, we then construct the solution for the cluster LP by computing the ratio from the pool of appropriately sampled clusters.
Here, we provide the procedure for sampling the single cluster (Algorithm~\ref{alg:sampling}).
Let $\varepsilon_\text{rt}=\varepsilon_1^2=\varepsilon^6$.
\begin{algorithm}
\caption{Sampling Single Cluster}
\label{alg:sampling}
\begin{algorithmic}
\State \textbf{Input:} Optimal bounded sub-cluster LP solution $\{y_S^{s,c}\}$
\State \textbf{Output:} Cluster $C$ of color $c$\State
\State Choose $(s,c,u)$ with probability proportional to $\frac{y_u^{s,c}}{s}$
\State Define $y_S'=\frac{y_{Su}^{s,c}}{y_{u}^{s,c}}$
\State $C\leftarrow \textsc{Raghavendra-Tan}(y')$
\end{algorithmic}
\end{algorithm}
\begin{lemma}[\cite{RT12}, Chromatic Extension of Lemma 16 of \citep{Cao2024}]
    There exists an algorithm that samples $C\subseteq V$ such that
\begin{itemize}
    \item (doesn't break atom) $K\in \Phi^p,\,K\cap C\neq \emptyset\implies K\subseteq C$,
    \item $\Pr[v\in C]=y_v',\,\forall v\in V$,
    \item $\frac{1}{|N_u^{2;c}\backslash K_u|^2}\sum_{v,w\in N_u^{2;c}\backslash K_u}|\Pr[v,w\in C]-y_{vw}'|\leq \varepsilon_\textnormal{rt}=\Theta\left(\frac{1}{\sqrt{r}}\right)$
\end{itemize}
    in $n^{O(r)}$ time.
\end{lemma}
We define $\err_{vw\mid u}^{s,c}$ as $|\Pr[v,w\in C]-y_{vw}'|$ conditioned on the choice of $(s,c,u)$. Any of the omitted $s,\,c,$ or $u$ in $\err_{vw}$ indicates the expectation of $\err_{vw\mid u}^{s,c}$ over those variables.
Note that $\Pr[v,w\in C]=0$ if $v\not\in N_u^{2;c}$ or $u\not\in N_u^{2;c}$, thus $\err_{vw\mid u}^{s,c}=0$ in this case.
\begin{proposition}
    $\sum_{(s,c,u)}\frac{y_{u}^{s,c}}{s}=y_\emptyset$.
\end{proposition}
\begin{proof}
    Plug $S=\emptyset$ into condition~\eqref{cond:sLP_CCC_trans}.
\end{proof}
\begin{lemma}[Chromatic Extension of Lemma 17 of \citep{Cao2024}]
    $\Pr[v\in C]=\frac{1}{y_\emptyset}$.
\end{lemma}
\begin{proof}
    \[\Pr[v\in C]=\frac{\sum_{(s,c,u)}\frac{y_{u}^{s,c}}{s}\cdot \frac{y_{uv}^{s,c}}{y_{u}^{s,c}}}{\sum_{(s,c,u)}\frac{y_{u}^{s,c}}{s}}=\frac{\sum_{(s,c,u)}\frac{y_{uv}^{s,c}}{s}}{y_\emptyset}=\frac{\sum_{(s,c)}y_{v}^{s,c}}{y_\emptyset}=\frac{1}{y_\emptyset}.\]
\end{proof}
\begin{lemma}[Chromatic Extension of Lemma 18 of \citep{Cao2024}]
    \label{lem:err}
    For any $vw\in \binom{V}{2}$,
    \begin{enumerate}
        \item $\Pr[v\in C,\,w\in C\mid c]\leq \frac{y_{vw}^c}{y_\emptyset^c}+\err_{vw}^c$,
        \item $\Pr[v\in C,\,w\in C]\leq \frac{y_{vw}}{y_\emptyset}+\err_{vw}$.
        \item $\Pr[v\in C]-\Pr[v\in C,\,w\in C\textnormal{ of color }c]\leq \frac{1-y_{vw}^c}{y_\emptyset}+\frac{y_\emptyset^c}{y_\emptyset}\err_{vw}^c$
    \end{enumerate}
\end{lemma}
\begin{proof}
    We first prove the first statement.
    \[
        \Pr[v\in C,\,w\in C\mid c]\leq \frac{1}{y_\emptyset^c}\sum_{(s,u)}\frac{y_{u}^{s,c}}{s}\cdot \left(\frac{y_{uvw}^{s,c}}{y_u^{s,c}}+\err_{vw\mid u}^{s,c}\right)=\frac{1}{y_\emptyset^c}\sum_{(s,u)}\frac{y_{uvw}^{s,c}}{s}+\err_{vw}^c=\frac{y_{vw}^c}{y_\emptyset^c}+\err_{vw}^c.
    \]
    The second statement then follows directly by considering all $c\in L$.
    For the third statement,
    \[\Pr[v\in C,\,w\in C\mid c]\geq \frac{1}{y_\emptyset^c}\sum_{(s,u)}\frac{y_{u}^{s,c}}{s}\cdot \left(\frac{y_{uvw}^{s,c}}{y_u^{s,c}}-\err_{vw\mid u}^{s,c}\right)=\frac{y_{vw}^c}{y_\emptyset^c}-\err_{vw}^c,\]
    hence
    \[\Pr[v\in C]-\Pr[v\in C,\,w\in C\text{ of color }c]=\frac{1}{y_\emptyset}-\frac{y_\emptyset^c}{y_\emptyset}\Pr[v\in C,\,w\in C\mid c]\leq \frac{1-y_{vw}^c}{y_\emptyset}+\frac{y_\emptyset^c}{y_\emptyset}\err_{vw}^c.\]
\end{proof}
The following lemma provides a slightly shorter proof of the original, as well as a generalization to the chromatic setting.
\begin{lemma}[Chromatic Extension of Lemma 19 of \citep{Cao2024}]
    $\sum_{vw\in\binom{V}{2}}\err_{vw}^c\leq O(\varepsilon_1)\cdot \frac{1}{y_{\emptyset}^c}|E^{2;c}|$.
\end{lemma}
\begin{proof}
    \begin{align*}
        \sum_{vw\in\binom{V}{2}}\err_{vw}^c
        &=\sum_{vw\in\binom{V}{2}}\frac{1}{y_\emptyset^c}\sum_{(s,u)}\frac{y_u^{s,c}}{s}\cdot \err_{vw\mid u}^c\\
        &=\frac{1}{y_\emptyset^c}\sum_{(s,u)}\frac{y_u^{s,c}}{s}\cdot \sum_{vw\in\binom{V}{2}}\err_{vw\mid u}^c\\
        &=\frac{1}{y_\emptyset^c}\sum_{u\in V}\sum_{s\in(|K_u|+\varepsilon_1|N_u^{2;c}\backslash K_u|, r]}\frac{y_u^{s,c}}{s}\cdot \sum_{vw\in\binom{V}{2}}\err_{vw\mid u}^c\\
        &\leq \frac{1}{y_\emptyset^c}\sum_{u\in V}\sum_{s\in(|K_u|+\varepsilon_1|N_u^{2;c}\backslash K_u|, r]}\frac{y_u^{s,c}}{s}\cdot \frac{1}{2}\varepsilon_\text{rt}|N_u^{2;c}\backslash K_u|^2.
    \end{align*}
    As in the proof in the original paper,
    \[s>|K_u|+\varepsilon_1|N_u^{2;c}\backslash K_u|\implies \varepsilon_1\cdot |N_u^{2;c}\backslash K_u|\leq s-|K_u|=\sum_{v\in N_u^{2;c}\backslash K_u}\frac{y_{uv}^{s,c}}{y_u^{s,c}},\]
    hence
    \begin{align*}
        \sum_{vw\in\binom{V}{2}}\err_{vw}^c
        &\leq \frac{1}{y_\emptyset^c}\sum_{u\in V}\sum_{s\in(|K_u|+\varepsilon_1|N_u^{2;c}\backslash K_u|, r]}\frac{y_u^{s,c}}{s}\cdot\frac{1}{2}\varepsilon_1|N_u^{2;c}\backslash K_u|\sum_{v\in N_u^{2;c}\backslash K_u}\frac{y_{uv}^{s,c}}{y_u^{s,c}}\\
        &= \frac{1}{2}\varepsilon_1\cdot \frac{1}{y_\emptyset^c}\sum_{uw\in E^{2;c}}\left(\sum_{s\in(|K_u|+\varepsilon_1|N_u^{2;c}\backslash K_u|, r]}\sum_{v\in N_u^{2;c}\backslash K_u}\frac{y_{uv}^{s,c}}{s}\right.\\
        &\quad +\left.\sum_{s\in(|K_w|+\varepsilon_1|N_w^{2;c}\backslash K_w|, r]}\sum_{v\in N_w^{2;c}\backslash K_w}\frac{y_{wv}^{s,c}}{s}\right)\\
        &\leq \frac{1}{2}\varepsilon_1\cdot \frac{1}{y_\emptyset^c}\sum_{uw\in E^{2;c}}(y_u^c+y_w^c)\leq \varepsilon_1\cdot \frac{1}{y_\emptyset^c}|E^{2;c}|.
    \end{align*}
\end{proof}
\subsection{Construction of Near-Optimal Solution}
\label{subsec:clp}
Finally, we provide a construction of a solution to the \eqref{eq:CLP-CCC} from the cluster sample pool (Algorithm~\ref{alg:solution}) and prove that the constructed solution is near-optimal, concluding the proof of Theorem~\ref{thm:solution}.
Let $\Delta=\Theta\left(\frac{n^2L\log n}{\varepsilon_1^2|E^2|}\right)$ such that $\Delta y_{\emptyset}$ is an integer.
\begin{algorithm}
    \caption{Constructing Chromatic Cluster LP solution}
    \label{alg:solution}
    \begin{algorithmic}
        \State \textbf{Input:} Optimal bounded sub-cluster LP solution $\{y_S^{s,c}\}$
        \State \textbf{Output:} Near-optimal \eqref{eq:CLP-CCC} solution $\{z_S^c\}$\State
        \State $R_u^c\leftarrow \emptyset$ for all $u\in V,\,c\in L$
        \For{$i=1,2,\ldots,\Delta y_\emptyset$}
            \State Sample $C_i$ of color $c_i$ using Algorithm~\ref{alg:sampling}
            \For{$u\in C_i$}
            \State Add $i$ to $R_u^{c_i}$
            \EndFor
        \EndFor
        \For{$u\in V$}
            \If{$|R_u|<\lceil(1-\varepsilon)\Delta \rceil$} Abort
            \EndIf
            \For{$(c,i)\in R_u$ s.t.\ $|R_u\cap [0,i)|\geq \lceil(1-\varepsilon)\Delta \rceil$ (descending order in $i$)}
                \State Remove $u$ from $C_i$
                \State Remove $i$ from $R_u^c\subseteq R_u$
            \EndFor
        \EndFor
        \State $z_S^c=\frac{1}{\lceil(1-\varepsilon)\Delta \rceil}|\{i:C_i=S\text{ of color }c\}|$ for all $S\subseteq V,\,S\neq \emptyset$
    \end{algorithmic}
\end{algorithm}
Here, $R_u:=\biguplus_{c\in L}R_u^c$, where each element is formulated as $(c,i),\,i\in R_u^c$.

Using Chernoff bound and union bound, we can prove the following:
\begin{lemma}[Chromatic Extension of application of Chernoff bound in~\cite{Cao2024}]
\label{lem:chernoff}
    With probability at least $1-1/n$, all the following hold before the removal process.
    \begin{itemize}
        \item $|R_u|\geq (1-\varepsilon_1)\Delta\geq (1-\varepsilon)\Delta$ for all $u\in V$,
        \item $|R_u\backslash (R_u^{\phi(uv)}\cap R_v^{\phi(uv)})|\leq (1+\varepsilon_1)\Delta(1-y_{uv}^{\phi(uv)}+y_\emptyset^{\phi(uv)}\err_{uv}^{\phi(uv)})+\frac{\varepsilon_1\Delta|E^2|}{n^2L}$ for all $uv\in \binom{V}{2},\,\phi(uv)\neq \gamma$,
        \item $|R_u\cap R_v|\leq (1+\varepsilon_1)\Delta(y_{uv}+y_\emptyset\err_{uv})+\frac{\varepsilon_1\Delta|E^2|}{n^2L}$ for all $uv\in \binom{V}{2},\,\phi(uv)=\gamma$.
    \end{itemize}
\end{lemma}
\begin{proof}
    Using the Chernoff bound \[\Pr[X<(1-\delta)\mu]<e^{-\delta^2\mu/2}\] with $n=\Delta y_\emptyset,\,\delta=\varepsilon_1,\,\mu=\Delta$ results in the violation in the first statement with at most $1/\text{poly}(n)$ probability.
    Note that the factor $L$ of $\Delta$ is required to ensure $\varepsilon_1^2\Delta=\Theta\left(\frac{n^2L\log n}{|E^2|}\right)=\Omega(\log n)$.

    Using Lemma~\ref{lem:err} and variant of the Chernoff bound \[\Pr[X>\mu+\delta\mu']<e^{-\delta^2\mu'/3}\;(\mu'\geq \mu)\] with $n=\Delta y_\emptyset,\,\delta=\varepsilon_1,\,
    \mu=\Delta y_\emptyset(\Pr[v\in C]-\Pr[v\in C,\,w\in C\text{ of color }c])
    \leq \Delta y_\emptyset\cdot\left(\frac{1-y_{uv}^c}{y_\emptyset}+\frac{y_\emptyset^c}{y_\emptyset}\err_{uv}^c\right),\,
    \mu'=\Delta y_\emptyset\cdot \max\left\{\frac{1-y_{uv}^c}{y_\emptyset}+\frac{y_\emptyset^c}{y_\emptyset}\err_{uv}^c,\frac{|E^2|}{y_\emptyset n^2L}\right\}\geq \mu$ results in the violation in the second statement with a probability of at most $1/\text{poly}(n)$.

    Similarly, using Lemma~\ref{lem:err} and variant of the Chernoff bound \[\Pr[X>\mu+\delta\mu']<e^{-\delta^2\mu'/3}\;(\mu'\geq \mu)\] with
    $n=\Delta y_\emptyset,\,\delta=\varepsilon_1,\,
    \mu=\Delta y_\emptyset\cdot \Pr[u\in C,\,v\in C]
    \leq \Delta y_\emptyset\cdot\left(\frac{y_{uv}}{y_\emptyset}+\err_{uv}\right),\,
    \mu'=\Delta y_\emptyset\cdot \max\left\{\frac{y_{uv}}{y_\emptyset}+\err_{uv},\frac{|E^2|}{y_\emptyset n^2L}\right\}
    \geq \mu$
    results in the violation in the third statement with a probability of at most $1/\text{poly}(n)$.

    Since the total number of statements is polynomial in $n$, scaling the constant factor of $\Delta$ appropriately bounds the probability of violation in at least one statement above by $1/n$ due to the union bound.

\end{proof}

After the removal process, each $R_u$ now has a size of $\lceil (1-\varepsilon)\Delta\rceil$.
Moreover, each value of $|R_u\backslash (R_u^{\phi(uv)}\cap R_v^{\phi(uv)})|$ and $|R_u\cap R_v|$ still remain in their bound in Lemma~\ref{lem:chernoff}.
For $|R_u\cap R_v|$, it is obvious since the value monotonically decreases during removal;
for $|R_u\backslash (R_u^{\phi(uv)}\cap R_v^{\phi(uv)})|$, the inequality holds since we remove indices from the largest side for all $u\in V$, as provided in the proof of the following lemma.

\begin{lemma}[Chromatic Extension of Claim 21 of \citep{Cao2024}]
    All the following hold after the removal process.
    \begin{itemize}
        \item $|R_u|= \lceil(1-\varepsilon)\Delta\rceil$ for all $u\in V$,
        \item $|R_u\backslash (R_u^{\phi(uv)}\cap R_v^{\phi(uv)})|\leq (1+\varepsilon_1)\Delta(1-y_{uv}^{\phi(uv)}+y_\emptyset^{\phi(uv)}\err_{uv}^{\phi(uv)})+\frac{\varepsilon_1\Delta|E^2|}{n^2L}$ for all $uv\in \binom{V}{2},\,\phi(uv)\neq \gamma$,
        \item $|R_u\cap R_v|\leq (1+\varepsilon_1)\Delta(y_{uv}+y_\emptyset\err_{uv})+\frac{\varepsilon_1\Delta|E^2|}{n^2L}$ for all $uv\in \binom{V}{2},\,\phi(uv)=\gamma$.
    \end{itemize}
\end{lemma}
\begin{proof}
    It is sufficient to prove the second statement.
    We examine the local change in the value during the specific choice in the order of removal,
    since the eventual change in value does not depend on the choice of the order of removal.
    Note that the value increases if and only if $i\in R_u^{\phi(uv)}\cap R_v^{\phi(uv)}$ is removed from $R_v$;
    the value decreases if and only if $i \in R_u\backslash (R_u^{\phi(uv)}\cap R_v^{\phi(uv)})$ is removed from $R_u$.
    
    \emph{W.l.o.g.}\ $|R_u|\geq |R_v|$ and thus $|R_u\backslash (R_u^{\phi(uv)}\cap R_v^{\phi(uv)})|\geq |R_v\backslash (R_u^{\phi(uv)}\cap R_v^{\phi(uv)})|$.
    Removing indices from $|R_u|$ until $|R_u|=|R_v|$, the value of the larger one monotonically decreases and eventually becomes identical to the smaller one.
    Note that the value of the smaller one might increase during the removal, but the value is still bounded above by the larger one.

    Now remove $\max R_u$ from $R_u$ and $\max R_v$ from $R_v$, until $|R_u|=|R_v|=\lceil(1-\varepsilon)\Delta\rceil$.
    According to the previous observation, the only way to increase $|R_u\backslash (R_u^{\phi(uv)}\cap R_v^{\phi(uv)})|$ after two removals is $\max R_u\in R_u^{\phi(uv)}\cap R_v^{\phi(uv)}$ followed by $\max R_v \in R_u^{\phi(uv)}\cap R_v^{\phi(uv)}$;
    however, this is not possible since the removed index $\max R_u=\max R_u^{\phi(uv)}\cap R_v^{\phi(uv)}$ after the first removal is located in $R_v\backslash (R_u^{\phi(uv)}\cap R_v^{\phi(uv)})$ and is larger than the following $\max R_u^{\phi(uv)}\cap R_v^{\phi(uv)}$.
    Therefore, the value monotonically decreases until the removal stage is completed.
    This finishes the proof.
\end{proof}

Finally, we show that the resulting $\{z_S^c\}$ is a feasible and near-optimal solution for the \eqref{eq:CLP-CCC}.
To check the feasibility, it is sufficient to check condition~\eqref{cond:CLP_v}, which is true due to the normalization by $|R_u|=\lceil(1-\varepsilon)\Delta\rceil$.
To check the near-optimality, we can check it by
\[x_{uv}^{\phi(uv)}=\frac{|R_u\backslash (R_u^{\phi(uv)}\cap R_v^{\phi(uv)})|}{\lceil(1-\varepsilon)\Delta\rceil}\leq \frac{1+\varepsilon_1}{1-\varepsilon}(1-y_{uv}^{\phi(uv)}+y_\emptyset^{\phi(uv)}\err_{uv}^{\phi(uv)})+\frac{\varepsilon_1}{1-\varepsilon}\frac{|E^2|}{n^2L}\]
for $uv\in \binom{V}{2},\,\phi(uv)\neq \gamma$,
\[\sum_{c\in L}(1-x_{uv}^c)=\frac{|R_u\cap R_v|}{\lceil(1-\varepsilon)\Delta\rceil}\leq \frac{1+\varepsilon_1}{1-\varepsilon}(y_{uv}+y_\emptyset\err_{uv})+\frac{\varepsilon_1}{1-\varepsilon}\frac{|E^2|}{n^2L}\]
for $uv\in \binom{V}{2},\,\phi(uv)= \gamma$, hence
\begin{align*}\obj(z)
&\leq (1+O(\varepsilon))\left(\opt+O(\varepsilon_1)\cdot \sum_{c\in L}|E^{2;c}|\right)+O(\varepsilon_1)\frac{|E^{2}|}{L}\\
&=(1+O(\varepsilon))\opt + O(\varepsilon_1)|E^2|\\
&=(1+O(\varepsilon))\opt + O(\varepsilon_1)O\left(\frac{1}{\varepsilon^2}\right)\opt\\
&=(1+O(\varepsilon))\opt.
\end{align*}

This finalizes the proof of Theorem~\ref{thm:solution}.

\acks{DL and CF were partially supported by CF’s New Faculty Startup Fund from Seoul National University (SNU).}

\bibliographystyle{plainnat}
\bibliography{references}

\begin{thebibliography}{30}
\providecommand{\natexlab}[1]{#1}
\providecommand{\url}[1]{\texttt{#1}}
\expandafter\ifx\csname urlstyle\endcsname\relax
  \providecommand{\doi}[1]{doi: #1}\else
  \providecommand{\doi}{doi: \begingroup \urlstyle{rm}\Url}\fi

\bibitem[Ailon et~al.(2008)Ailon, Charikar, and Newman]{ailon2008aggregating}
N.~Ailon, M.~Charikar, and A.~Newman.
\newblock Aggregating inconsistent information: Ranking and clustering.
\newblock volume~55, New York, NY, USA, November 2008. Association for Computing Machinery.
\newblock \doi{10.1145/1411509.1411513}.
\newblock URL \url{https://doi.org/10.1145/1411509.1411513}.

\bibitem[Anava et~al.(2015)Anava, Avigdor-Elgrabli, and Gamzu]{anava2015chromatic}
Y.~Anava, N.~Avigdor-Elgrabli, and I.~Gamzu.
\newblock Improved theoretical and practical guarantees for chromatic correlation clustering.
\newblock In \emph{Proceedings of the 24th International Conference on World Wide Web}, WWW '15, page 55–65, Republic and Canton of Geneva, CHE, 2015. International World Wide Web Conferences Steering Committee.
\newblock ISBN 9781450334693.
\newblock \doi{10.1145/2736277.2741629}.
\newblock URL \url{https://doi.org/10.1145/2736277.2741629}.

\bibitem[Assadi and Wang(2022)]{assadi2022sublinear}
S.~Assadi and C.~Wang.
\newblock Sublinear time and space algorithms for correlation clustering via sparse-dense decompositions.
\newblock In \emph{Proceedings of the 13th Innovations in Theoretical Computer Science Conference (ITCS)}, pages 10:1--10:20, 2022.

\bibitem[Bansal et~al.(2004)Bansal, Blum, and Chawla]{bansal2004correlation}
N.~Bansal, A.~Blum, and S.~Chawla.
\newblock Correlation clustering.
\newblock \emph{Machine Learning}, 56\penalty0 (1-3):\penalty0 89--113, 2004.
\newblock ISSN 0885-6125.
\newblock \doi{10.1023/B:MACH.0000033116.57574.95}.

\bibitem[Bateni et~al.(2023)Bateni, Esfandiari, Fichtenberger, Henzinger, Jayaram, Mirrokni, and Wiese]{bateni2023}
M.~Bateni, H.~Esfandiari, H.~Fichtenberger, M.~Henzinger, R.~Jayaram, V.~Mirrokni, and A.~Wiese.
\newblock Optimal fully dynamic k-center clustering for adaptive and oblivious adversaries.
\newblock In \emph{Proceedings of the 2023 Annual ACM-SIAM Symposium on Discrete Algorithms (SODA)}, pages 2677--2727, 2023.

\bibitem[Behnezhad et~al.(2022)Behnezhad, Charikar, Ma, and Tan]{behnezhad2022}
S.~Behnezhad, M.~Charikar, W.~Ma, and L.~Tan.
\newblock Almost 3-approximate correlation clustering in constant rounds.
\newblock In \emph{63rd IEEE Annual Symposium on Foundations of Computer Science (FOCS)}, pages 720--731, 2022.

\bibitem[Behnezhad et~al.(2023)Behnezhad, Charikar, Ma, and Tan]{behnezhad2023}
S.~Behnezhad, M.~Charikar, W.~Ma, and L.~Tan.
\newblock Single-pass streaming algorithms for correlation clustering.
\newblock In \emph{Proceedings of the 2023 ACM-SIAM Symposium on Discrete Algorithms (SODA)}, pages 819--849, 2023.

\bibitem[Ben-Dor and Yakhini(1999)]{ben-dor1999clustering}
A.~Ben-Dor and Z.~Yakhini.
\newblock Clustering gene expression patterns.
\newblock In \emph{Proceedings of the Third Annual International Conference on Computational Molecular Biology}, RECOMB '99, page 33–42, New York, NY, USA, 1999. Association for Computing Machinery.
\newblock ISBN 1581130694.
\newblock \doi{10.1145/299432.299448}.
\newblock URL \url{https://doi.org/10.1145/299432.299448}.

\bibitem[Bonchi et~al.(2015)Bonchi, Gionis, Gullo, Tsourakakis, and Ukkonen]{bonchi2012chromatic}
F.~Bonchi, A.~Gionis, F.~Gullo, C.~E. Tsourakakis, and A.~Ukkonen.
\newblock Chromatic correlation clustering.
\newblock \emph{ACM Trans. Knowl. Discov. Data}, 9\penalty0 (4), June 2015.
\newblock ISSN 1556-4681.
\newblock \doi{10.1145/2728170}.
\newblock URL \url{https://doi.org/10.1145/2728170}.

\bibitem[Braverman et~al.(2025)Braverman, Dharangutte, Pai, and Shah]{braverman2025fully}
V.~Braverman, P.~Dharangutte, S.~Pai, and V.~Shah.
\newblock Fully dynamic adversarially robust correlation clustering in polylogarithmic update time.
\newblock In \emph{Proceedings of the 28th International Conference on Artificial Intelligence and Statistics (AISTATS)}, 2025.

\bibitem[Cao et~al.(2024)Cao, Cohen-Addad, Lee, Li, Newman, and Vogl]{Cao2024}
N.~Cao, V.~Cohen-Addad, E.~Lee, S.~Li, A.~Newman, and L.~Vogl.
\newblock Understanding the cluster linear program for correlation clustering.
\newblock In \emph{Proceedings of the 56th Annual ACM Symposium on Theory of Computing}, STOC 2024, page 1605–1616, New York, NY, USA, 2024. Association for Computing Machinery.
\newblock ISBN 9798400703836.
\newblock \doi{10.1145/3618260.3649749}.
\newblock URL \url{https://doi.org/10.1145/3618260.3649749}.

\bibitem[Cao et~al.(2025{\natexlab{a}})Cao, Cohen-Addad, Lee, Li, Lolck, Newman, Thorup, Vogl, Yan, and Zhang]{Cao2025}
N.~Cao, V.~Cohen-Addad, E.~Lee, S.~Li, D.~R. Lolck, A.~Newman, M.~Thorup, L.~Vogl, S.~Yan, and H.~Zhang.
\newblock Solving the correlation cluster lp in sublinear time.
\newblock In \emph{Proceedings of the 57th Annual ACM Symposium on Theory of Computing}, STOC '25, page 1154–1165, New York, NY, USA, 2025{\natexlab{a}}. Association for Computing Machinery.
\newblock ISBN 9798400715105.
\newblock \doi{10.1145/3717823.3718181}.
\newblock URL \url{https://doi.org/10.1145/3717823.3718181}.

\bibitem[Cao et~al.(2025{\natexlab{b}})Cao, Cohen{-}Addad, Lee, Li, Lolck, Newman, Thorup, Vogl, Yan, and Zhang]{abs-2504-12060}
N.~Cao, V.~Cohen{-}Addad, E.~Lee, S.~Li, D.~Rasmussen Lolck, A.~Newman, M.~Thorup, L.~Vogl, S.~Yan, and H.~Zhang.
\newblock Static to dynamic correlation clustering.
\newblock \emph{CoRR}, abs/2504.12060, 2025{\natexlab{b}}.

\bibitem[Charikar et~al.(2005)Charikar, Guruswami, and Wirth]{CHARIKAR2005360}
M.~Charikar, V.~Guruswami, and A.~Wirth.
\newblock Clustering with qualitative information.
\newblock \emph{Journal of Computer and System Sciences}, 71\penalty0 (3):\penalty0 360--383, 2005.
\newblock ISSN 0022-0000.
\newblock \doi{https://doi.org/10.1016/j.jcss.2004.10.012}.
\newblock URL \url{https://www.sciencedirect.com/science/article/pii/S0022000004001424}.
\newblock Learning Theory 2003.

\bibitem[Chawla et~al.(2015)Chawla, Makarychev, Schramm, and Yaroslavtsev]{chawla2015near}
S.~Chawla, K.~Makarychev, T.~Schramm, and G.~Yaroslavtsev.
\newblock Near optimal lp rounding algorithm for correlation clustering on complete and complete k-partite graphs.
\newblock In \emph{Proceedings of the 47th Annual ACM Symposium on Theory of Computing (STOC)}, pages 219--228, 2015.

\bibitem[Cohen-Addad et~al.(2019)Cohen-Addad, Hjuler, Parotsidis, Saulpic, and Schwiegelshohn]{cohen2019fully}
V.~Cohen-Addad, N.~Hjuler, N.~Parotsidis, D.~Saulpic, and C.~Schwiegelshohn.
\newblock Fully dynamic consistent facility location.
\newblock In \emph{Proceedings of the 33rd Annual Conference on Neural Information Processing Systems (NeurIPS)}, 2019.

\bibitem[Cohen-Addad et~al.(2021)Cohen-Addad, Lattanzi, Mitrovic, Norouzi-Fard, Parotsidis, and Tarnawski]{lattanzi2021parallel}
V.~Cohen-Addad, S.~Lattanzi, S.~Mitrovic, A.~Norouzi-Fard, N.~Parotsidis, and J.~Tarnawski.
\newblock Correlation clustering in constant many parallel rounds.
\newblock In \emph{Proceedings of the 38th International Conference on Machine Learning (ICML)}, volume 139, pages 2069--2078, 2021.

\bibitem[Cohen-Addad et~al.(2022{\natexlab{a}})Cohen-Addad, Lattanzi, Maggiori, and Parotsidis]{cohen2022online}
V.~Cohen-Addad, S.~Lattanzi, A.~Maggiori, and N.~Parotsidis.
\newblock Online and consistent correlation clustering.
\newblock In \emph{Proceedings of the 39th International Conference on Machine Learning (ICML)}, pages 4157--4179, 2022{\natexlab{a}}.

\bibitem[Cohen-Addad et~al.(2022{\natexlab{b}})Cohen-Addad, Lee, and Newman]{cohenaddad2022sherali}
V.~Cohen-Addad, E.~Lee, and A.~Newman.
\newblock { Correlation Clustering with Sherali-Adams }.
\newblock In \emph{2022 IEEE 63rd Annual Symposium on Foundations of Computer Science (FOCS)}, pages 651--661, Los Alamitos, CA, USA, November 2022{\natexlab{b}}. IEEE Computer Society.
\newblock \doi{10.1109/FOCS54457.2022.00068}.
\newblock URL \url{https://doi.ieeecomputersociety.org/10.1109/FOCS54457.2022.00068}.

\bibitem[Cohen-Addad et~al.(2023)Cohen-Addad, Lee, Li, and Newman]{cohenaddad2023preclustering}
V.~Cohen-Addad, E.~Lee, S.~Li, and A.~Newman.
\newblock Handling correlated rounding error via preclustering: A 1.73-approximation for correlation clustering.
\newblock In \emph{2023 IEEE 64th Annual Symposium on Foundations of Computer Science (FOCS)}, pages 1082--1104, 2023.
\newblock \doi{10.1109/FOCS57990.2023.00065}.

\bibitem[Cohen-Addad et~al.(2024)Cohen-Addad, Lattanzi, Maggiori, and Parotsidis]{cohen2024dynamic}
V.~Cohen-Addad, S.~Lattanzi, A.~Maggiori, and N.~Parotsidis.
\newblock Dynamic correlation clustering in sublinear update time.
\newblock In \emph{Proceedings of the 41st International Conference on Machine Learning (ICML)}, pages 9230--9270, 2024.

\bibitem[Dalirrooyfard et~al.(2024)Dalirrooyfard, Makarychev, and Mitrovic]{abs-2402-15668}
M.~Dalirrooyfard, K.~Makarychev, and S.~Mitrovic.
\newblock Pruned pivot: Correlation clustering algorithm for dynamic, parallel, and local computation models.
\newblock \emph{CoRR}, abs/2402.15668, 2024.
\newblock URL \url{https://doi.org/10.48550/arXiv.2402.15668}.

\bibitem[Fichtenberger et~al.(2021)Fichtenberger, Mirrokni, and Zadimoghaddam]{fichtenberger2021}
H.~Fichtenberger, V.~Mirrokni, and M.~Zadimoghaddam.
\newblock Correlation clustering in data streams.
\newblock In \emph{Proceedings of the 33rd Annual ACM Symposium on Parallelism in Algorithms and Architectures (SPAA)}, 2021.

\bibitem[Guo et~al.(2021)Guo, Mitrovic, and Vassilvitskii]{guo2021distributed}
A.~Guo, S.~Mitrovic, and S.~Vassilvitskii.
\newblock Distributed correlation clustering.
\newblock In \emph{Advances in Neural Information Processing Systems (NeurIPS)}, 2021.

\bibitem[Jaghargh et~al.(2019)Jaghargh, Vassilvitskii, and Lattanzi]{jaghargh2019}
M.~Jaghargh, S.~Vassilvitskii, and S.~Lattanzi.
\newblock Scalable correlation clustering: An empirical study.
\newblock In \emph{Proceedings of the 25th ACM SIGKDD International Conference on Knowledge Discovery \& Data Mining}, 2019.

\bibitem[Klodt et~al.(2021)Klodt, Seifert, Zahn, Casel, Issac, and Friedrich]{klodt2021color}
N.~Klodt, L.~Seifert, A.~Zahn, K.~Casel, D.~Issac, and T.~Friedrich.
\newblock A color-blind 3-approximation for chromatic correlation clustering and improved heuristics.
\newblock In \emph{Proceedings of the 27th ACM SIGKDD Conference on Knowledge Discovery \& Data Mining}, KDD '21, page 882–891, New York, NY, USA, 2021. Association for Computing Machinery.
\newblock ISBN 9781450383325.
\newblock \doi{10.1145/3447548.3467446}.
\newblock URL \url{https://doi.org/10.1145/3447548.3467446}.

\bibitem[Lattanzi and Vassilvitskii(2017)]{lattanzi2017consistent}
S.~Lattanzi and S.~Vassilvitskii.
\newblock Consistent k-clustering.
\newblock In \emph{Proceedings of the 34th International Conference on Machine Learning (ICML)}, pages 1975--1984, 2017.

\bibitem[Lee et~al.(2025)Lee, Fan, and Lee]{Lee2025}
D.~Lee, C.~Fan, and E.~Lee.
\newblock Improved approximation algorithms for chromatic and pseudometric-weighted correlation clustering.
\newblock \emph{CoRR}, abs/2505.21939, 2025.
\newblock \doi{10.48550/ARXIV.2505.21939}.
\newblock URL \url{https://doi.org/10.48550/arXiv.2505.21939}.

\bibitem[Raghavendra and Tan()]{RT12}
P.~Raghavendra and N.~Tan.
\newblock \emph{Approximating CSPs with Global Cardinality Constraints Using SDP Hierarchies}, pages 373--387.
\newblock \doi{10.1137/1.9781611973099.33}.
\newblock URL \url{https://epubs.siam.org/doi/abs/10.1137/1.9781611973099.33}.

\bibitem[Xiu et~al.(2022)Xiu, Han, Tang, Cui, and Huang]{XiuHTCH22}
Q.~Xiu, K.~Han, J.~Tang, S.~Cui, and H.~Huang.
\newblock Chromatic correlation clustering, revisited.
\newblock In \emph{Advances in Neural Information Processing Systems 35: Annual Conference on Neural Information Processing Systems 2022, NeurIPS 2022, New Orleans, LA, USA, November 28 - December 9}, 2022.

\end{thebibliography}
\end{document}